\documentclass[conference]{IEEEtran}

\usepackage{amsmath}
\usepackage{booktabs}
\usepackage{enumitem}
\usepackage{amssymb}
\usepackage{amsthm}
\usepackage{framed}
\usepackage{graphicx}
\usepackage{multirow}
\usepackage{url} 
\usepackage{xcolor}
\usepackage{placeins}
\usepackage{subcaption}
\usepackage{pifont}
\usepackage{marvosym}
\usepackage{float}

\newtheoremstyle{mydef}
  {3pt}
  {3pt}
  {}
  {}
  {\bfseries}
  {.}
  {1em}
  {\textbf{Definition \thedefinition.} \quad \textit{#1}}  

\theoremstyle{mydef}
\newtheorem{definition}{Semi-honest Secure}  
\renewcommand{\thedefinition}{\arabic{definition}}

\newtheoremstyle{mythe}
  {3pt}      
  {3pt}      
  {}         
  {}         
  {\bfseries}
  {.}        
  {1em}      
  {\textbf{Theorem \thetheorem}}  

\theoremstyle{mythe}
\newtheorem{theorem}{}

\ifCLASSINFOpdf
\else
\fi
\begin{document}
%
\title{Practical Traceable Over-Threshold Multi-Party Private Set Intersection}

\author{\IEEEauthorblockN{Le Yang\IEEEauthorrefmark{1},
Weijing You\IEEEauthorrefmark{2}\textsuperscript{\Letter},
Huiyang He\IEEEauthorrefmark{1}, 
Kailiang Ji\IEEEauthorrefmark{3} and
Jingqiang Lin\IEEEauthorrefmark{1}\textsuperscript{\Letter} }
\IEEEauthorblockA{\IEEEauthorrefmark{1}School of Cyber Science and Technology, University of Science and Technology of China}
\IEEEauthorblockA{\IEEEauthorrefmark{2}Fujian Provincial Key Laboratory of Network Security and Cryptology,\\College of Computer and Cyber Security, Fujian Normal University}
\IEEEauthorblockA{\IEEEauthorrefmark{3}NIO Inc.}
\{leyang, hhe\}@mail.ustc.edu.cn, youweijing@fjnu.edu.cn, \\kailiang.ji@nio.com, linjq@ustc.edu.cn}
	

%


\maketitle

\begin{abstract}
Multi-Party Private Set Intersection (MP-PSI) with threshold enhances the flexibility of MP-PSI by disclosing elements present in at least $t$ participants' sets, rather than requiring elements to appear in all $n$ sets.
In scenarios where each participant is responsible for its dataset, e.g., digital forensics, MP-PSI with threshold is expected to disclose both intersection elements and corresponding holders such that elements are traceable and hence the reliability of intersection is guaranteed. We refer to MP-PSI with threshold supporting traceability as Traceable Over-Threshold Multi-Party Private Set Intersection (T-OT-MP-PSI). However, research on such protocols remains limited, and the current solution is resistant to $t-2$ semi-honest participants at the cost of considerable computational overhead.

In this paper, we propose two novel Traceable OT-MP-PSI protocols. The first protocol is the \underline{E}fficient \underline{T}raceable OT-MP-PSI (ET-OT-MP-PSI), which combines Shamir's secret sharing with oblivious programmable pseudorandom function, achieving significantly improved efficiency with resistance to at most $t-2$ semi-honest participants.
The second one is the \underline{S}ecurity-enhanced \underline{T}raceable OT-MP-PSI (ST-OT-MP-PSI), which achieves security against up to $n-1$ semi-honest participants by further leveraging oblivious linear evaluation protocol. 

Compared to the recent Traceable OT-MP-PSI protocol by Mahdavi et al., our protocols eliminate the security assumption that certain special parties do not collude and provide stronger security guarantees.
We implemented our proposed protocols and conducted extensive experiments under various settings. 
We compared the performance of our protocols with that of  Mahdavi et al.'s protocol.
While our Traceable OT-MP-PSI protocols enhance security, experimental results demonstrate high efficiency. For instance, given 5 participants, with threshold of 3, and set sizes are $2^{14}$, our  ET-OT-MP-PSI protocol is 15056$\times$ faster, and the ST-OT-MP-PSI is 505$\times$ faster, compared to Mahdavi et al.'s protocol.
\end{abstract}


%
\IEEEpeerreviewmaketitle

\section{Introduction}

Multi-Party PSI (MP-PSI) allows three or more participants, each holding a private set, to learn nothing but the intersection of their sets. Currently, MP-PSI is widely used in privacy-sensitive domains, such as cache sharing in edge computing \cite{nguyen2022mpccache}, federated learning \cite{lu2020multi,elkordy2022federated}, and anomaly detection \cite{kolesnikov2017practical,inbar2018efficient,ghosh2019algebraic}. Implicitly, typical MP-PSI only identifies elements that are possessed by all participants. For participants interested in computing the intersection over elements held not necessarily by all, but by at least a predefined number of participants, we have the MP-PSI with threshold.

Intuitively, MP-PSI with threshold should provide ideal privacy protection, and hence most existing MP-PSI with threshold protocols \cite{kissner2004private, bay2021practical,chandran2021efficient} are fully anonymous, i.e., they reveal nothing but the intersection itself. However, in certain scenarios, full anonymity may be problematic:

\begin{itemize}[leftmargin=*]
         \item \textit{Network Anomaly Detection:} In distributed environments, anomaly detection systems deployed at different nodes independently report suspicious behaviors. 
Previous works \cite{valdes2001probabilistic,lippmann20001999} have demonstrated that the traceability of alerts across distributed detectors significantly improves anomaly attribution and operational response. However, conventional MP-PSI with threshold lacks such a capability, thereby impeding accurate localization and coordinated mitigation.
        
        \item \textit{Digital Forensics Investigation:} In digital forensics, the evidence is distributed among multiple entities. Studies \cite{li2023anonymous,selamat2011traceability,selamat2013forensic} emphasize that establishing the provenance of digital evidence is essential for constructing reliable evidence chains, ensuring admissibility in court, and coordinating multi-agency investigations.
 Although conventional MP-PSI with threshold typically lacks traceability, this prevents investigators from identifying which parties hold the evidence.
        
        \item \textit{Suspicious Account Analysis:} Taking anti-money laundering (AML) as an example, Kings Research projects that the global AML market will reach \$9.692 billion by 2031 \cite{kingresearch2024aml}. As part of practical AML initiatives, the Hong Kong Monetary Authority's AMLab leverages network analysis to identify mule accounts. However, the lack of traceability in fully anonymous MP-PSI with threshold hinders effective cross-institutional collaboration, which is explicitly emphasized as a critical requirement in AML guidelines \cite{chatain2009preventing}.
\end{itemize}

We observe that, when participants are accountable for their private sets, which are contributed to the community decision, traceability for each intersection element becomes necessary. Therefore, Over-Threshold MP-PSI with traceability was introduced \cite{mahdavi2020practical}, extending  conventional MP-PSI with threshold by not only identifying elements appearing in at least $t$ participants' sets, but also disclosing the identities of the parties holding each intersecting element. However, the existing solution \cite{mahdavi2020practical} remains limited in both efficiency and security.
Specifically, in terms of security, it only resists collusion among up to $t-2$ semi-honest participants under the assumption that certain special parties do not collude.
In terms of efficiency, its computational complexity is $O(m(n\log(\frac{m}{t}))^{2t})$, which grows exponentially with the threshold $t$, where $n$ is the total number of parties and $m$ is the set size.
Furthermore, empirical results demonstrate poor performance in practice. For example, when the number of participants is $n=10$, the threshold is $t=7$, and the set size is $2^5$, the runtime of their protocol exceeds 9 hours in our experiments.


In this paper, we refer to this functionality as  \textbf{\underline{T}raceable \underline{O}ver-\underline{T}hreshold \underline{M}ulti-\underline{P}arty \underline{P}rivate \underline{S}et \underline{I}ntersection (T-OT-MP-PSI)} and propose two protocol instantiations: the \underline{E}fficient \underline{T}raceable OT-MP-PSI (ET-OT-MP-PSI) and the \underline{S}ecurity-enhanced \underline{T}raceable OT-MP-PSI (ST-OT-MP-PSI). The ET-OT-MP-PSI achieves high efficiency with resistance to $t-2$ collusion among participants, while the ST-OT-MP-PSI is secure against collusion by up to $n-1$ participants in the semi-honest adversary model. Here, $t$ denotes the predefined threshold for intersection computation, and $n$ denotes the total number of participants.

\subsection{ The high-level idea of our protocols }

To enable a Traceable OT-MP-PSI, we first consider two basic functionalities: revealing elements and corresponding holders only when the number of holders is at least the threshold. We adopt Shamir's secret sharing, which is also utilized in Mahdavi et al. \cite{mahdavi2020practical}. In our design, each intersection element is treated as ``secret'' to be shared and can be reconstructed when enough ``shares'' from participants are collected. 
To preserve privacy, the intersection elements should not be revealed during such a secret sharing process. Therefore, we  incorporate the Oblivious Programmable Pseudorandom Function (OPPRF) to transmit ``shares''. The OPPRF ensures that only participants with the same element receive the correct ``share'', whereas others receive random values.

Each participant may have elements that are not in the final intersection. Note that in the MP-PSI with threshold, only elements possessed by enough participants are revealed, while those appearing in fewer sets should remain private. 
However, given Shamir's secret sharing and OPPRF, it is easy to identify elements in common by simply observing the ``shares'' from different participants, even if those elements do not belong to the intersection.
This problem was addressed by Mahdavi et al. \cite{mahdavi2020practical} using homomorphic encryption, which results in high computational cost. In this work, inspired by Herzberg et al. \cite{herzberg1995proactive}, we update the element ``shares'' with additional shares of zero-value secret which does not modify the underlying secret nor the threshold value.

\noindent\textbf{Efficient Traceable OT-MP-PSI.} Based on Shamir's secret sharing, OPPRF and secret shares update, we first propose an Efficient Traceable OT-MP-PSI (ET-OT-MP-PSI) protocol. It is initialized by one participant, who then interacts with the other participants to derive the final intersection in the following three phases: 1) \textit{conditional secret sharing,} in which elements in the set are transformed into secret shares and securely transmitted among all participants using OPPRF; 2) \textit{secret shares update,} in which each participant individually updates their own shares using new shares of zero-value secret before collection; 3) \textit{conditional collection and reconstruction,} in which OPPRF is used again by the leader participant to collect shares, and the elements are reconstructed to obtain the intersection and corresponding holders. 

ET-OT-MP-PSI is an extension of MP-PSI (CCS'17) \cite{kolesnikov2017practical} to MP-PSI with threshold. However, based on the Lagrange Interpolation theorem, this extension is only secure against collusion among up to $t-2$ participants. Nevertheless, for stricter security requirements in practice, \textit{is it possible to remove the dependence between security and the threshold?}

\noindent\textbf{Security-enhanced Traceable OT-MP-PSI.} In the ET-OT-MP-PSI, participants are able to access shares directly derived from threshold Shamir's secret sharing. Consequently, if more than $t-2$ corrupted participants collude, they can infer the private information of honest parties, thereby compromising the security of the protocol. To address this, we introduce the Oblivious Linear Evaluation (OLE) protocol to enable a three-party interaction during the shares update phase, thereby further imposing requirements on the elements held by participants. That is, successful reconstruction requires not only collecting enough shares, but also ensuring that a sufficient number of parties hold the same element. Hence the resistance of our Traceable OT-MP-PSI to semi-honest adversaries is extended from $t-2$ to $n-1$.

\subsection{Our Contributions}
    Our contributions can be summarized as follows:
    \begin{itemize}[leftmargin=*]
        \item \textbf{Efficient Traceable OT-MP-PSI.} We revisit full anonymity and traceability of MP-PSI with threshold from a practical perspective, and propose an Efficient Traceable OT-MP-PSI protocol (ET-OT-MP-PSI) based on threshold Shamir's secret sharing and OPPRF. The ET-OT-MP-PSI is efficient and resistant to collusion among up to $t-2$ semi-honest participants.
        
        \item \textbf{Security-enhanced Traceable OT-MP-PSI.} We introduce a Security-enhanced Traceable OT-MP-PSI protocol (ST-OT-MP-PSI), which incorporates the OLE protocol into the ET-OT-MP-PSI. This protocol could tolerate at most $n-1$ semi-honest adversaries at a modest performance cost.
        
        \item \textbf{Security analysis, implementation and performance evaluation.} 
        We conduct a security analysis of the two proposed protocols. In addition, we implement and evaluate them under various experimental settings. Experimental results demonstrate that both of our protocols outperform the recent work with similar functionality.

    \end{itemize}

\section{Preliminaries}

\subsection{Notations}
\begin{table}[H]
\renewcommand{\arraystretch}{0.95}
    \centering
    \label{tab:notations}
    
    \begin{tabular}{@{\centering}p{0.2\columnwidth}p{0.66\columnwidth}@{}}
    \toprule
    \multicolumn{1}{c}{{\textbf{Notations}}} & \multicolumn{1}{c}{\textbf{Descriptions}} \\ 
    \midrule
    \multicolumn{1}{c}{$n$} & The number of parties \\ 
    \multicolumn{1}{c}{$m$} & Set size of each party \\ 
    \multicolumn{1}{c}{$t$} & The threshold value \\ 
    \multicolumn{1}{c}{$[c,d]$ or $[a]$} & Denotes the set $\{c,c+1,...,d\}$ or $\{0,1,...,a\}$ \\
    \multicolumn{1}{c}{$P_i$} & The party with index $i,i \in [n]$ \\
    \multicolumn{1}{c}{$S_i$} & The set of party $P_i$, i.e., $S_i=\{e_0^i,\dots,e_{m-1}^i\}$ \\
    \multicolumn{1}{c}{$m_b$} & The size of the Simple or Cuckoo hashing table \\
    \multicolumn{1}{c}{$\lambda$} & The statistical security parameter \\
    \multicolumn{1}{c}{$\kappa$} & The computational security parameter \\
    \multicolumn{1}{c}{$B_S[b]$ or $B_C[b]$} & The  $b^{\text{th}}$ bin of Simple or Cuckoo hashing table\\
    \multicolumn{1}{c}{$x\leftarrow \mathbb{F}_p$} & $x$ is sampled uniformly over the field $\mathbb{
    F}_p$\\

    \bottomrule
    \end{tabular}
    \vspace{-0.3cm}
    \end{table}

\subsection{Security Model}
\textcolor{black}{
Consistent with the majority of the prior research on MP-PSI \cite{kolesnikov2017practical,inbar2018efficient,bay2021practical,miyaji2015scalable,wu2024ring}, our proposed protocols primarily focus on the \textit{semi-honest} adversarial model \cite{lindell2017simulate}. In this model, adversaries may try to learn as much information as possible from the protocol execution but will not deviate from the execution steps. These adversaries are also referred to as honest-but-curious. 
}

\textcolor{black}{
The view of a party consists of its private input, its random tape , and the list of all messages received during the protocol. The view of the adversary comprises the combined views of all corrupted parties, potentially allowing multiple colluding parties to aggregate their information and infer private data.
To prove the security of a protocol, it is common to construct a simulator $\mathsf{Sim}$ that, given the adversaries' inputs $X$ and outputs $Y$, generates simulated views  that are computationally indistinguishable from the adversaries' real views in the protocol execution \cite{kolesnikov2017practical,gao2024efficient,wu2024ring,nevo2021simple}.
}
\textcolor{black}{
\begin{definition}
     \textit{A protocol $\pi$ securely realizes functionality $ \mathcal{F} $ in the presence of semi-honest adversaries if there exists a simulator $ \mathsf{Sim} $ such that, for any subset of corrupt parties $\{P_i\in \mathbb{C}\}$, the views $\{\mathsf{Sim}(X,Y,\mathbb{C})\}$ generated by the simulator are computationally indistinguishable from the views $\{\mathsf{view}^\pi_\mathbb{\mathbb{C}}(X,Y)\}$ obtained by the adversaries in the real execution. Formally, this can be expressed as:
     $$
        \{\mathsf{Sim}(X,Y,\mathbb{C})\}\overset{c}{\equiv} \{\mathsf{view}^\pi_\mathbb{\mathbb{C}}(X,Y)\}.
     $$
     }
\end{definition}\label{Definition1}
}

\subsection{Shamir's Secret Sharing}

In the $(t,n)$-Shamir's secret sharing scheme, the dealer distributes the secret $S$ to $n$ participants, with each participant possessing a share of the secret. When $t$ or more participants collaborate, they   reconstruct the secret $S$ together. If fewer than $t$ participants are involved, they will not gain any information about the secret.

\subsubsection{Secret sharing and Reconstruction}
Shamir's secret sharing \cite{shamir1979share} is a $(t,n)$-threshold secret sharing scheme. During the secret distribution phase, the dealer chooses a prime number $p$ and randomly picks $t-1$ numbers $a_i,i \in [t-1]$ from the finite field  $\mathbb{F}_p$ to construct polynomial with secret $S$ :
\begin{equation*}
    f(x)=S+a_1x+\dots+a_{t-1}x^{t-1}.
\end{equation*}

For each participant $i$, the dealer evaluates the polynomial and distributes the secret share $(x_i, y_i)$, \textcolor{black}{where $y_i = f(x_i)$ and $x_i \neq 0$.
Without loss of generality, and for the sake of clarity, we define $x_i=i+1$ throughout the protocol.}
This scheme uses Lagrange interpolation theorem. Specifically, $t$ points on the polynomial can uniquely determine a polynomial with degree equal to or less than $t-1$.  Given any $t$ secret shares, secret $S$ can be reconstructed by  Lagrange interpolation:
\begin{equation*}
    S=f(0)=\sum_{i=0}^{t-1}y_i \prod_{j=0,j \neq i}^{t-1}(\frac{x_j}{x_j-x_i}).
\end{equation*}

\subsubsection{Secret shares update} \label{sub:Secret shares updating}
To update the secret shares held by participants, each participant  generates a random polynomial $f'(x)$ with a constant term of 0, expressed as
\begin{equation*}
    f'(x)=0+b_1x+\dots+b_{t-1}x^{t-1}.
\end{equation*}
Using this polynomial, the participant computes an update share $(x_i,y_i')$ for each participant $i$, where $y_i'=f'(x_i)$ and $x_i \neq 0$. Similarly, for consistency and ease of presentation, we define $x_i=i+1$ in our protocol.
Upon receiving this update share, each participant updates their secret share by setting $(x_i,y_i+y_i')$ as the new share. Essentially, this process performs a Shamir's secret sharing with a secret value of 0, ensuring that the correctness of reconstruction remains intact while updating the secret shares. This approach is similar to the proactive secret sharing scheme proposed by Herzberg et al. in \cite{herzberg1995proactive}.

\subsection{Hashing Schemes} \label{sub:hashing schemes}
\subsubsection{Simple Hashing}
In Simple hashing, the hash table consists of $m_b$ bins $B[0],\dots, B[m_b-1]$. \textcolor{black}{By uniformly selecting a hash function $h:\{0,1\}^* \rightarrow [0,m_b-1]$ at random,} element $e$ is mapped to bin $B[h(e)]$ in the hash table and inserted to hash table by appending it to this bin. Obviously, there are multiple elements in the same bin. 
In practice, multiple hash functions are often employed to reduce the probability of collision and improve load balancing.

\subsubsection{Cuckoo Hashing}
Cuckoo hashing scheme uses hash function $h_1,\dots,h_k:\{0,1\}* \rightarrow [m_b]$ to map $m$ elements to $m_b$ bins in hash table. Unlike Simple hashing, Cuckoo hashing is only allowed to store one element per bin. One variant of Cuckoo hashing is Cuckoo hashing with a stash. To insert an element $e$ into hash table do the following \cite{pinkas2018scalable}: (1) If one of bin $B[h_1(e)],\dots,B[h_k(e)]$ is empty, insert element $e$ into the empty bin. (2) Otherwise, the element $e$ is inserted into the bin $B[h_1(e)]$, evicting its existing content $o$. The evicted element $o$ is then relocated to a new bin $B[h_i(o)]$, using $h_i$ to determine the new bin location, where $h_i(o) \neq h_1(e) $ for $i \in [1,\dots,k]$.
The procedure is repeated until no more evictions are necessary, or until a threshold number of relocations been performed. In the latter case, the last element is placed in a stash. After Cuckoo hashing, element $ e$ can be found in the one of  following locations: bin $B[h_1(e)],\dots,B[h_k(e)]$ or stash.

Another variant of Cuckoo hashing, as proposed in \cite{kolesnikov2017practical}, eliminates the use of a stash by employing two hash tables. This design avoids the inefficiencies associated with a stash, where every item in one party’s stash need to be compared to every item of another party, increasing overhead. Specifically, the procedure starts by using three ``primary'' Cuckoo hash functions to determine the placement of an element. If these initial attempts are unsuccessful, the process resorts to two ``supplementary'' Cuckoo hash functions as a fallback mechanism.
By adjusting the parameters within the hashing scheme, the process can be ensured to succeed with a negligible failure probability, specifically less than $ 2^{-\lambda} $.

\subsection{Oblivious Programmable Pseudorandom Function} \label{sub:opprf}
Oblivious pseudorandom function (OPRF) \cite{kolesnikov2016efficient} is a two-party protocol through which the sender learns a pseudo-random function (PRF) key $k$, and the receiver learns $F(k,q_1),\dots,F(k,q_v)$, where $F$ is a pseudo-random function and  $(q_1,\dots,q_v)$ are the receiver's inputs. 
If the sender's \textcolor{black}{input} $ x $ matches the receiver's input $ q_i $, the sender can compute $ F(key, x) $, which equals $ F(\text{key}, q_i) $, under the key $ k $.

Oblivious programmable pseudorandom function (OPPRF) \cite{kolesnikov2017practical} is similar to OPRF, with the additional property that on a certain programmed set of inputs the function outputs programmed values. In the OPPRF, the sender inputs a set of points $\{(x_1,y_1),\dots,(x_u,y_u)\}$, and the receiver inputs $(q_1,\dots,q_v)$. By running protocol, the receiver ultimately obtains OPPRF output $(hint, F(k,hint,q_1), \dots, F(k,hint,q_v))$ and the sender gets $(k,hint)$.
Within the protocol, when the receiver's input $ q_i $ equals the sender's input $ x_j $ (i.e., $ q_i = x_j $), the receiver is able to obtain the value $ y_j $, which has been programmed by the sender. 
For the receiver, it is indistinguishable whether the obtained output is a random value or a value programmed by the sender; meanwhile, the sender remains oblivious to the receiver’s input.
The  functionality of OPPRF is presented in Fig. \ref{fig:OPPRF Functionality}.

Kolesnikov et al. proposed three instantiation methods for OPPRF, among which the table-based construction has favorable communication and computational cost \cite{kolesnikov2017practical}. However, this construction allows the receiver to evaluate the programmable PRF on only $v = 1$ point. Therefore, they extended this construction using hashing schemes to support both a large $u$ (the number of programmed points) and a large $v$ (the number of queries). Specifically, in the OPPRF protocol, the sender uses Simple hashing to map $x_i, i \in [1,u]$ into $m_b$ bins, while the receiver uses Cuckoo hashing to map $q_j, j \in[1,v]$ into $m_b$ bins. Now in each bin, the receiver has at most one item $q$. Therefore, they can run the table-based OPPRF protocol on these inputs. They refer to this protocol as the hashing-based OPPRF protocol.

\begin{figure}
    \centering
     \makebox[\columnwidth]{
        \fbox{
        \parbox{\dimexpr\columnwidth-2\fboxsep-2\fboxrule}
        {
            \textsc{Parameters}: 
                \begin{itemize}
                    \item A \textcolor{black}{programmable} pseudorandom function $F$.
                    \item The upper bound $u$ on the number of points to be programmed.
                    \item The bound $v$ on the number of queries.
                \end{itemize}
                
            \textsc{Input}: 
                \begin{itemize}
                    \item The sender inputs  points $\{(x_1,y_1),\dots,(x_u,y_u)\}$.
                    \item The receiver inputs $(q_1,\dots,q_v)$.
                \end{itemize}

            \textsc{Output}: 
                \begin{itemize}
                     \item For each $q_i$, if $ q_i = x_j $, the receiver obtains $ y_j $; otherwise, the receiver receives a random value.
                \end{itemize}
            
        }
    }
    }
    \caption{The OPPRF functionality.}
    \label{fig:OPPRF Functionality}   
    \vspace{-0.5cm}
\end{figure}

\subsection{Oblivious Linear Evaluation}
Oblivious linear evaluation (OLE) is a two-party protocol and serves as a fundamental building block in multi-party secure computation protocols \cite{rindal2021vole}. 
In the OLE protocol, the sender inputs  $a$  and  $b$, where  $a$  and  $b$  are elements of a finite field  $\mathbb{F}$ . The receiver inputs $x \in \mathbb{F}$ and ultimately receives $y \in \mathbb{F}$ such that
$
y = ax + b
$.
Throughout the protocol, the sender remains oblivious to the receiver's input  $x$, and the receiver does not learn any information about the sender's inputs $a$  and  $b$.

Vector OLE (VOLE) is the vectorized variant of the OLE protocol, allowing the receiver to learn a linear combination of two vectors held by the sender. Specifically, the sender inputs vectors $\boldsymbol \alpha, \boldsymbol \beta \in \mathbb{F}^n$, while the receiver inputs $x \in \mathbb{F}$ and obtains $\boldsymbol y \in \mathbb{F}^n$ where
$
\boldsymbol y = \boldsymbol \alpha x + \boldsymbol \beta
$.

Batch OLE (BOLE) is similar to VOLE but extends it by allowing the receiver's input to also be a vector. Specifically, the receiver inputs a vector $ \boldsymbol x \in \mathbb{F}^n $ and obtains  $\boldsymbol y \in \mathbb{F}^n$ where
$
\boldsymbol y = \boldsymbol \alpha \boldsymbol x + \boldsymbol \beta
$.

\section{Traceable Over-Threshold Multi-Party Private Set Intersection}
\subsection{Functionality Definition}
Before introducing our proposed protocols, we first define the functionality of protocols. The protocols require $ n \geq 3$ parties, denoted as $P_0,\dots,P_{n-1}$,  each holding a private set of size $m$, denoted as $S_0,\dots,S_{n-1}$, along with a threshold value $t$. 
The ultimate goal of the protocols is to enable party $ P_0 $ to obtain the following information: 
\begin{itemize}[leftmargin=*]
    \item \textbf{Intersection elements:} Party $P_0$ identifies each element $e_i$ from its own set $S_0$ that satisfy the threshold condition $c_i \geq t$, where $t$ is the predefined threshold value and $c_i$ denotes the number of parties holding the element $e_i$. 
    \item \textbf{Identity of element holders:} The protocol reveals the specific participants $\{P_j\}$ holding each intersection element.
    \item \textbf{Counting each intersection element:} Since the identities of the holders are disclosed, $P_0$ naturally infers the number of each intersection element. 
\end{itemize}

We refer to such protocols as the \textbf{Traceable Over-Threshold Multi-Party Private Set Intersection}.
The ideal functionality $\mathcal{F}_{\text{T-OT-MP-PSI}}^{n,m,t}
$ is formally described in Fig. \ref{fig:Traceable OT-MP-PSI Functionality}.
\begin{figure}
    \centering
     \makebox[\columnwidth]{
        \fbox{
        \parbox{\dimexpr\columnwidth-2\fboxsep-2\fboxrule}
        {
            \textsc{Parameters}: 
                \begin{itemize}
                    \item $n \geq 3$ parties $P_0,\dots,P_{n-1}$, with their respective private sets ${S}_0,\dots,{S}_{n-1}$ of size $m$.
                \end{itemize}
                
            \textsc{Input}: 
                \begin{itemize}
                    \item Each party $P_i$ has a private set $S_i$ as input.
                    \item A threshold value $t$, where $1 < t\leq n$.
                \end{itemize}

            \textsc{Output}: 
                \begin{itemize}
                     \item $P_0$ outputs the intersection set $I=\{(e_i,c_i,\{P_j\})|,e_i \in S_0,c_i \geq t\}$, where $e_i$ is the intersection element, $c_i$ is the number of parties holding element $e_i$, and $\{P_j\}$ is the set of these parties.
                    \item $P_1,\dots,P_{n-1}$ outputs $\bot$.
                \end{itemize}
            
        }
    }
    }
    \caption{Traceable OT-MP-PSI functionality $\mathcal{F}_{\text{T-OT-MP-PSI}}^{n,m,t}$.}
    \label{fig:Traceable OT-MP-PSI Functionality}
    \vspace{-0.5cm}
\end{figure}

\subsection{Protocol Overview}

Motivated by the limitations of fully anonymous MP-PSI with threshold in regulatory scenarios, we aim to develop Traceable OT-MP-PSI. Considering the characteristics of the protocol, the first challenge arises:
\begin{framed}
\noindent \textit{\textbf{Challenge 1:} How can we design an MP-PSI with threshold that supports traceability? }
\end{framed}

Motivated by Mahdavi et al. (ACSAC'20) \cite{mahdavi2020practical}, we observe that the traceability in MP-PSI with threshold can be achieved by combining Shamir's secret sharing with the additively homomorphic Paillier cryptosystem. However, the traceable MP-PSI construction in \cite{mahdavi2020practical} still exhibits limitations in both security and performance. Specifically, even assuming no collusion across certain special parties, the protocol is secure against collusion among up to $t-2$ participants, where $t$ is the threshold. For performance, take 10 participants with a threshold of 7 and a set size of $2^5$ as an example. In this case, the protocol will consume more than 9 hours.
Therefore, we take a step back and are inspired by the work of Kolesnikov et al. (CCS'17) \cite{kolesnikov2017practical}. An approach to extending traditional MP-PSI to MP-PSI with threshold and traceability is to combine Shamir's secret sharing with Oblivious Programmable Pseudorandom Function (OPPRF). 

Specifically, each element in $P_0$'s set is processed into distinct shares and assigned to specific participants, enabling \textit{traceability}. 
To obtain the intersection, during the distribution and collection of shares, shares of each potential element should be received by target participants holding the same element, which is achieved through OPPRF. 
If the final collected shares for an element successfully reconstruct the original secret, it indicates that at least $t$ participants hold the element. This realizes the \textit{threshold functionality}, confirms its inclusion in the intersection and reveals corresponding holders.

Nevertheless, such a scarecrow Traceable OT-MP-PSI may expose the privacy of participants' sets. Since party $P_0$, which initializes the protocol, holds all the shares of the secret, it is able to infer whether other participants hold elements not in the intersection. This is achieved by simply comparing the initial shares it sent with the final shares it received, even if the process involves OPPRF. Such inference violates the general requirement for privacy protection in MP-PSI. So another challenge occurs:
\begin{framed}
\noindent \textit{\textbf{Challenge 2:} How can we preserve the privacy of non-intersecting elements while ensuring the functionality of the protocol?}
\end{framed}
Updating the secret shares held by each participant is an option. Specifically, we employ zero-value Shamir's secret sharing to update the shares held by each party, effectively refreshing the original secret shares to protect the privacy of participants and ensuring the correctness of secret reconstruction, as outlined in Section \ref{sub:Secret shares updating}. By leveraging Shamir's secret sharing to update secret shares, we achieve updating while avoiding additional costly operations, thereby proposing the Efficient Traceable OT-MP-PSI (ET-OT-MP-PSI).

In such an extension from MP-PSI \cite{kolesnikov2017practical} to MP-PSI with threshold that supports traceability, each party can directly obtain the correct values for updating their shares. As a result, any collusion of $ t-1 $ corrupted parties (including $ P_0 $) compromises the privacy of the honest parties. Consequently, the ET-OT-MP-PSI can only withstand collusion by up to $ t-2 $ parties. The overall process of the protocol is illustrated in Fig.~\ref{fig:protocol_process}.
\begin{figure}
    \centering
    \includegraphics[width=0.85\linewidth]{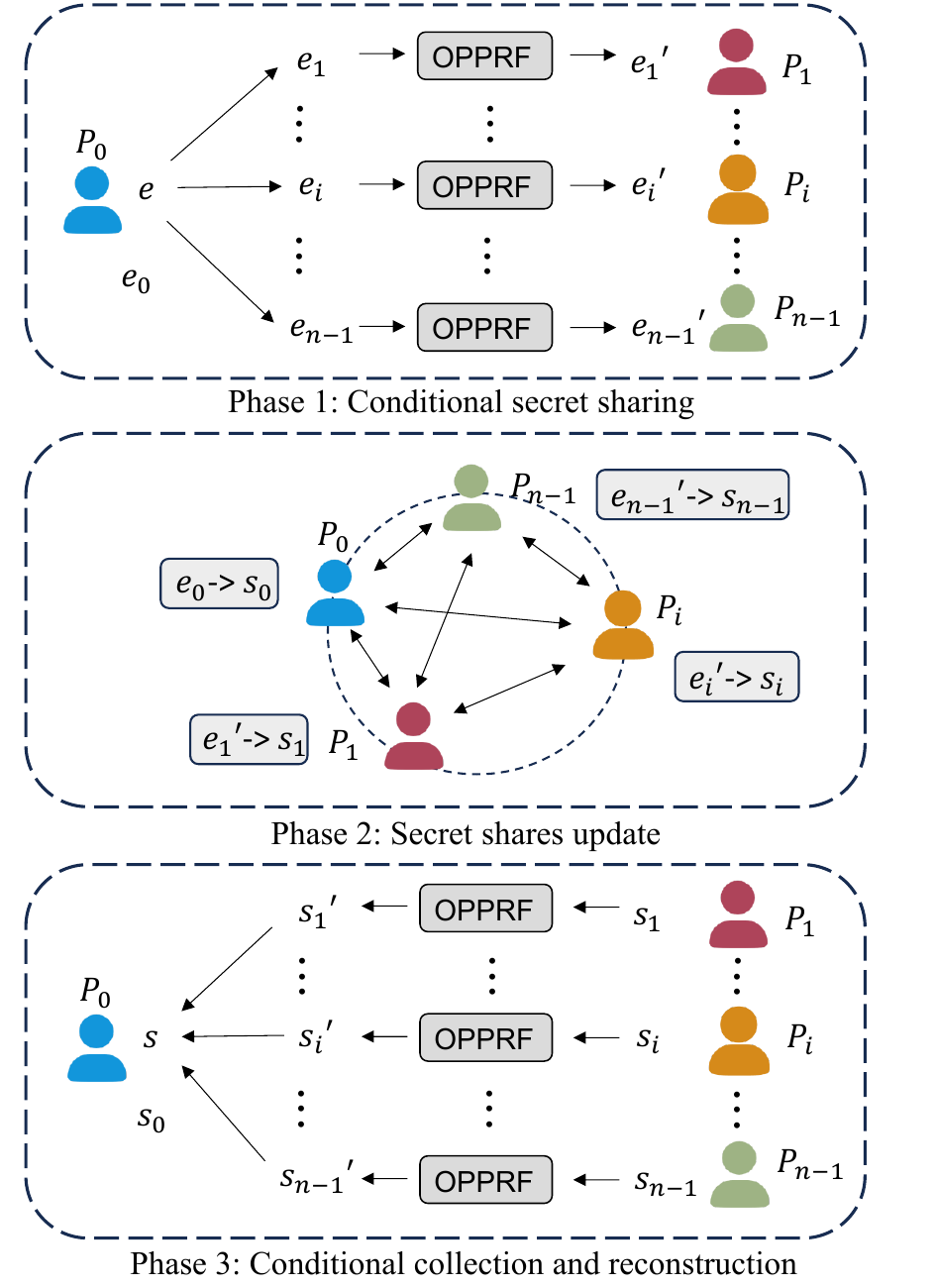}
    \caption{The overall process of ET-OT-MP-PSI.}
    \label{fig:protocol_process}
    \vspace{-0.5cm}
\end{figure}

To meet higher security requirements, we now take a step forward in the level of security. A Traceable OT-MP-PSI protocol should be robust against collusion among an arbitrary number of participants. We now face the final challenge: 

\begin{framed}
\noindent \textit{\textbf{Challenge 3:} How can we design a Traceable OT-MP-PSI that is secure against arbitrary collusion among semi-honest participants?}
\end{framed}

To enhance security, we introduce the OLE protocol to enable a three-party interaction during the shares update phase. 
{In this design, party $ P_i $ receives the correct values for updating shares from party $ P_j $ only if $ P_i $ holds the same element as party $ P_0 $.  
Hence, $t-1$ colluding parties can infer whether the honest party $P_i$ holds a specific element $e$ only if all of the colluding parties also hold the same element.
When $n-1$ parties collude, if fewer than $t-1$ of them possess the element $e$, the colluding parties are unable to compromise the protocol's security. Conversely, if at least $t-1$ of the colluding parties hold $e$, the protocol output reveals whether $e$ is part of the intersection and discloses its associated holders. Since this information is explicitly included in the output of the protocol, such inference does not result in any additional privacy leakage. Therefore, the Security-enhanced Traceable OT-MP-PSI is secure against arbitrary collusion among semi-honest participants.}

In conclusion, both of our proposed Traceable OT-MP-PSI protocols consist of the following three main phases:
\begin{itemize}[leftmargin=*]
    \item \textbf{Conditional secret sharing.} Elements in the set are transformed into secret shares and securely transmitted among all participants using OPPRF.
    \item \textbf{Secret shares update.} Each participant individually updates their own shares using new shares of zero-value secret before collection.
    \item \textbf{Conditional collection and reconstruction.} Use OPPRF again to collect shares and reconstruct elements to obtain the intersection and corresponding holders.
\end{itemize}

In the following sections, we present the Efficient Traceable OT-MP-PSI (ET-OT-MP-PSI) and the Security-enhanced Traceable OT-MP-PSI (ST-OT-MP-PSI) in detail.

\subsection{Details of the Efficient Traceable OT-MP-PSI } \label{sub:Details of the Efficient Traceable OT-MP-PSI}

We now present our first protocol, ET-OT-MP-PSI. A formal description of the protocol is in Fig.~\ref{fig:ET-OT-MP-PSI Protocol}. 
\textcolor{black}{Prior to the protocol execution, each participant maps their set using  both Cuckoo hashing and Simple hashing, as described in Section \ref{sub:hashing schemes}. When applying Cuckoo hashing to the set elements, if the $b^{\text{th}}$ bin is empty, it is padded with a dummy element. Similarly, when using Simple hashing, each bin is padded with dummy elements so that its total size reaches the maximum bin size $\beta$. we adopt the method by Kolesnikov et al.~\cite{kolesnikov2017practical} for calculating $\beta$.
The padding is performed to hide the number of elements that were mapped to a specific bin, which would leak information about the input.}

\begin{figure}[]
    \centering
    \makebox[\columnwidth]{    
    \fbox{
        \parbox{\dimexpr\columnwidth-2\fboxsep-2\fboxrule}
        {
            \textsc{Parameters}: 
            \begin{itemize}
                \item $n$ parties $P_0,\dots,P_{n-1}$.
                \item A prime $p$, used for $(t,n)$-Shamir’s secret sharing.
                \item Hash functions for hashing scheme.
            \end{itemize}

            \textsc{Input}:
            \begin{itemize}
                \item Each party $P_i$ uses its private set $S_i$ as input.
                \item The threshold $t$. 
            \end{itemize}

            \textsc{Protocol}:
            \\  $\underline{\mathit{Conditional \ Secret \ Sharing}} $
            \begin{itemize}
                \item [(1)] Each party maps its set $S_i$ into bins using Cuckoo and Simple hashing schemes, obtaining $B_C[\cdot]$ and $B_S[\cdot]$. \textcolor{black}{Each empty bin in $B_C[\cdot]$ is padded with a dummy element, while each bin in $B_S[\cdot]$ is padded with dummy elements to the maximum bin size $\beta$.}
                \item [(2)] For each element $e_k^0  \in S_0, k \in [m]$, $ P_0 $ treats it as a secret and performs $(t, n)$-Shamir's secret sharing, generating $ n $ shares $s_k^{0,0},\dots, s_k^{0,n-1}$.
                \item [(3)] For the $b^{\text{th}}$ bin, $b \in [m_b]$, $P_0$ invokes an OPPRF protocol with every other party $P_i$, where $i \in [1,n-1]$.
                \begin{itemize}
                    \item $P_0$ is the sender with input $\{(e_k^0,s_k^{0,i})|e_k^0 \in B_S[b]\}$.
                    \item $P_i$ acts as receiver with input $\{e_k^i|e_k^i \in B_C[b]\}$. As a result, for every $e_k^i \in S_i$, $P_i$ obtains a corresponding OPPRF output, denoted as $\hat{s}_k^{0,i}$.            
                \end{itemize}
            \end{itemize}
             $\underline{\mathit{Secret \ Shares \ Update}}$
            \begin{itemize}
                \item [(4)] For  the $b^{\text{th}}$ bin, each $P_i, i \in [1,n-1]$ performs secret shares update mentioned in Section \ref{sub:Secret shares updating} to get shares $(j+1,f_{i,b}(j+1))$ and directly sends to $P_j,j \in [n]$ .
                \item [(5)] Upon receiving the values sent by the other participants, for  the $b^{\text{th}}$ bin, party $P_j$  sums the $ n-1 $ values to obtain $ \delta_b = f_{1,b}(j+1) +\dots+f_{n-1,b}(j+1)$. For $e_k^0 \in B_C[b]$, $P_0$ updates its shares: $y_k^0 = s_k^{0,0}+\delta_b$.
            \end{itemize}
             $\underline{\mathit{Conditional\ Collection\ and\ Reconstruction}}$
            \begin{itemize}
                \item [(6)] For  the $b^{\text{th}}$ bin, each pair of $P_i$ and $P_0$ invokes an OPPRF.
                \begin{itemize}
                    \item  $P_i$ is the sender with input $\{(e_k^i,\mu_k^{0,i}) | e_k^i \in B_S[b]\}$, where $\mu_k^{0,i} = \hat{s}_k^{0,i} + \delta_b$. 
                    \item  $P_0$ is the receiver with input $\{e_k^0|e_k^0 \in B_C[b]\}$ and obtains $y_k^i$ which represents the corresponding OPPRF output.
                \end{itemize}
                \item [(7)]\textcolor{black}{For each element $e_k^0 \in S_0$, $P_0$ applies Lagrange interpolation over all subsets of $t$ shares among the $n$ values $y_k^i$, always including its own share, to compute $Recon(y_k^i) = f_k(\cdot)$. If $f_k(0) = e_k^0$, then $e_k^0$ is identified as an intersection element. Subsequently, $P_0$ determines all the holders of $e_k^0$ by checking whether $f_k(i+1) = y_k^i$ holds for each $i \in [1, n-1]$. Finally, $P_0$ obtains the intersection set $I$.}


            \end{itemize}
            
        }
    }
    }
    \caption{ET-OT-MP-PSI protocol.}
    \label{fig:ET-OT-MP-PSI Protocol}
    \vspace{-0.5cm}
\end{figure}

In the conditional secret sharing phase, for each element $e_k^0,k \in [m]$ in $S_0$, $P_0$ performs $(t, n)$-Shamir's secret sharing, generating $n$  shares $s_k^{0,0},\dots, s_k^{0,n-1}$.
Subsequently, $P_0$ executes the OPPRF protocol with each of the other parties $P_i, i\in[1,n-1]$, following Section \ref{sub:opprf}. $P_0$ programs the OPPRF using $\{(e_k^0,s_k^{0,i})|k \in [m]\}$, and $P_i$ acts as the receiver with input $S_i$. After the OPPRF is executed, each party $P_i$ obtains a corresponding OPPRF output for each $e_k^i$, denoted as $\hat{s}_k^{0,i}$. 
According to the properties of the OPPRF protocol, if $e_k^0 = e_k^i$, then $s_k^{0,i} = \hat{s}_k^{0,i}$. Moreover, $P_i$ does not know whether the received values are the real shares or random values.

In the shares update phase, $P_i,i\in[1,n-1]$ performs secret shares update as described in Section \ref{sub:Secret shares updating}. 
Specifically, for the the $b^{\text{th}}$ bin, $ P_i $ generates a random  polynomial with a constant term of $0$, a degree of $ t-1 $:
\begin{equation*}
    f_{i,b}(x)=0+a_1x+\dots+a_{t-1}x^{t-1}.
\end{equation*}
Subsequently, $P_i$ generates the corresponding secret shares $(j+1,f_{i,b}(j+1))$ for each party and directly sends the secret shares to the other parties  $P_j,j \in [n]$. For  the $b^{\text{th}}$ bin, $P_j$ sums $n-1$ obtained shares to obtain values $\delta_b$, for updating the original secret shares. Then, for $e_k^0 \in B_C[b]$, $P_0$ updates its shares: $y_k^0 = s_k^{0,0}+\delta_b$.

\textcolor{black}{
Finally, for each element $e_k^0$, $P_0$ enumerates all possible subsets of $t$ shares from the $n$ values and applies Lagrange interpolation to each subset to compute $Recon(y_k^{i}) = f_k(\cdot)$. This is equivalent to selecting all subsets of $t - 1$ shares from the remaining $n - 1$ values, since $P_0$'s own share is always correct and included in each reconstruction attempt. If any reconstruction satisfies $f_k(0) = e_k^0$, it indicates that $e_k^0$ is an intersection element. Furthermore, $P_0$ identifies all holders of this element by checking whether $f_k(i+1) = y_k^i$ holds for each $i \in [1, n - 1]$.
}

\subsection{Security-Enhanced Traceable OT-MP-PSI} \label{sub:Security-Enhanced Traceable OT-MP-PSI}

\subsubsection{Design Rationale}
In ET-OT-MP-PSI, collusion among $ t-1 $ participants (including $ P_0 $) reveals whether an honest party $ P_i $ holds a specific element $ e \in S_0 $ without secret reconstruction.

Specifically, Shamir's secret sharing relies on the Lagrange interpolation theorem, which states that a polynomial of degree at most $t-1 $ is uniquely determined by $t$  points. 
Now, suppose $ P_0, \dots, P_{t-2} $ collude, and $ P_i $ is the honest party. 
In the first phase, $ P_0 $ shares the element $ e $ as a secret using Shamir's secret sharing, generating a share $ x_i $ for $ P_i $. Then, $ P_i $ obtains the corresponding OPPRF output $ y_i $. The updated share, denoted as $y_i'$ in the second phase, is calculated as 
\begin{equation}
     y_i' = y_i + f_1(i+1) +  \dots + f_{n-1}(i+1).
\end{equation}
In Equation~$(1)$, $ f_1(i+1), \dots, f_{t-2}(i+1) $ are the values generated by the polynomials of parties $ P_1, \dots, P_{t-2} $. The polynomials $ f_{t-1}(\cdot), \dots, f_{n-1}(\cdot) $ can be reconstructed by the $ t-1 $ colluding parties $P_0,\dots,P_{t-2}$ according to the Lagrange interpolation theorem, enabling them to obtain the values $ f_{t-1}(i+1), \dots, f_{n-1}(i+1) $. 
Thus, in the final stage, after the OPPRF execution, $ P_0 $ determines whether the honest party $ P_i $ holds the element $ e $ by verifying if the OPPRF output satisfies:
\begin{equation}
    \text{OPPRF output} \stackrel{?}{=} x_i + f_1(i+1) + \dots + f_{n-1}(i+1).
\end{equation}
If Equation~$ (2) $ holds, it signifies party $ P_i $ possesses the element $ e $.

The root cause of such inference is that the values for updating the shares are directly sent to the parties, and reconstruction only relies on zero-value Shamir's secret sharing. Therefore, the security of the protocol can be compromised if $t-1$ corrupted parties collude, i.e., obtain enough shares to reconstruct the polynomial.
To prevent such inference, we further impose requirements on the elements held by participants, such that reconstruction is only feasible when a sufficient number of parties possess the same element. That is, successful reconstruction requires not only collecting enough shares but also ensuring that enough parties possess the identical element.
To implement the above idea, OLE is introduced to enable a three-party interaction during the shares update process. Specifically, $P_0$ and $P_i$ independently perform the OLE protocol with party $P_j$ such that the OLE outputs received by $P_0$ and $P_i$ collectively produce the correct value for updating $P_i$'s secret share, only if $P_0$ and $P_i$ hold the same element. Consequently, $P_i$'s share is updated correctly. 
As shown in Fig. \ref{fig:OLE_interaction}, if $ P_0 $ and $ P_i $ hold the same element, i.e., $ x = y $, the sum of their outputs equals $ f_j(i+1) $, which is the correct value for updating $ P_i $'s corresponding share.
By incorporating OLE, we develop the Security-enhanced Traceable OT-MP-PSI (ST-OT-MP-PSI).
In this protocol, when $n-1$ parties collude:
\begin{itemize}[leftmargin=*]
    \item If fewer than $t-1$ of them hold the element $e$,  the colluding parties cannot compromise the protocol's security through the inference described earlier.
    \item If at least $t-1$ of the colluding parties hold $e$, they can directly learn from the protocol output whether $e$ is in the intersection, and, if so, identify its holders. Since this information is explicitly included in the output of the protocol, such inference does not lead to any additional privacy leakage.
\end{itemize}
As a result, ST-OT-MP-PSI is secure against arbitrary collusion in the semi-honest model. The formal security proof is provided in Section~\ref{security analysis}.

\begin{figure}
    \centering
    \includegraphics[width=0.85\linewidth]{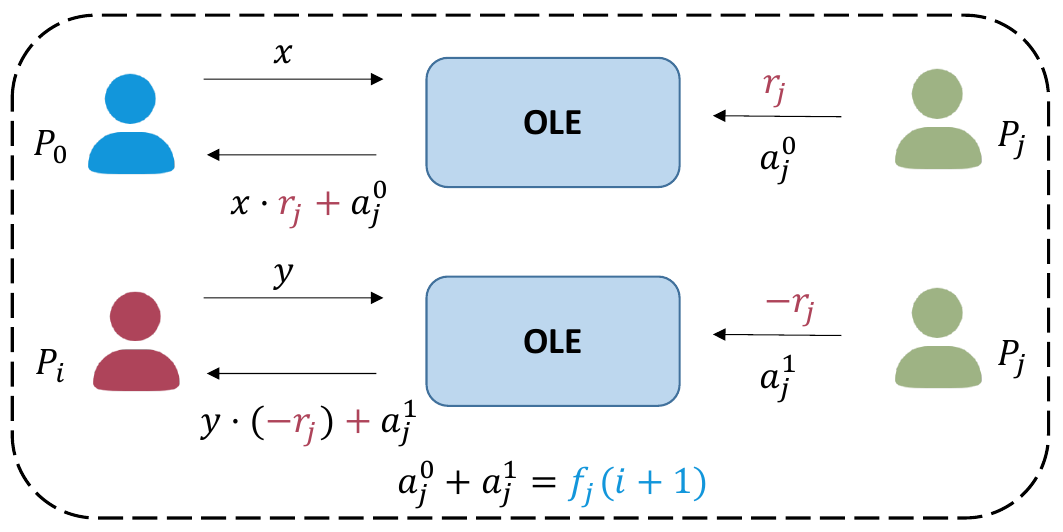}
    \caption{The core idea of the shares update phase in ST-OT-MP-PSI.}
    \label{fig:OLE_interaction}
    \vspace{-0.5cm}
\end{figure}

\subsubsection{Details of ST-OT-MP-PSI}

The ST-OT-MP-PSI follows the same basic steps as the ET-OT-MP-PSI, consisting of three phases, but there are some differences in certain aspects. The formal description of this protocol is given in Fig. \ref{fig:ST-OT-MP-PSI Protocol}.

At the outset, each party maps their set using Cuckoo hashing and Simple hashing, obtaining $B_C[\cdot]$ and $B_S[\cdot]$. 
\textcolor{black}{Similarly, to prevent information leakage, each party pads bins according to the corresponding hashing scheme. }

In the secret sharing phase, $ P_0 $ generates a random value $e_k^0{'}$ as the secret, which is uniquely mapped to an element $e_k^0$ in its set $S_0$, and performs $(t, n)$-Shamir's secret sharing to obtain $n$ shares $s_k^{0,0},\dots, s_k^{0,n-1}$. Subsequently, these shares are conditionally distributed to other parties through the OPPRF. The reason for not directly using $ P_0 $'s elements as secrets for secret sharing is to prevent $ t $ colluding parties from reconstructing with shares obtained through the OPPRF, which could potentially reveal information about $ P_0 $'s elements.

In the shares update process, for the $b^{\text{th}}$  bin, $ P_j, j \in [1,n-1]$ generates a polynomial $ f_{j,b}(\cdot) $ and directly sends the value $f_{j,b}(1)$ to $P_0$. Then, $P_0$ uses these values to directly update its own share. Afterward, both $ P_i $ and $ P_0 $, acting as receivers, execute the OLE protocol with each of the other parties $ P_j ,i,j\in[1,n-1]$. 
For each bin $ B_C[b] $, $ P_0 $ and $ P_j $ execute the OLE protocol $ \beta $ times, where $ \beta $ is the maximum bin size of $ B_S $. In each execution, $ P_j $, acting as the sender, inputs different $ r_j^v $ and $ a_j^{0,v},v \in [\beta]$, while $ P_0 $, acting as the receiver, inputs the element $ e_k^0$, which is located in its bin $B_C[b]$. As a result, $ P_0 $ obtains $ e_k^0 \cdot r_j^v + a_j^{0,v} $.
Similarly, for each bin $B_S[b]$, $P_i$ and $P_j$ invoke $\beta$ instances of the OLE protocol. In each execution, $P_j$, acting as the sender, inputs $-r_j^v$ and $a_j^{1,v}$, while $P_i$, acting as the receiver, inputs the element $e_k^i$ from its bin $B_S[b]$ and receives $e_k^i \cdot -r_j^v + a_j^{1,v}$. Here, $r_j^v$ is a random value, and  $a_j^{0,v} +a_j^{1,v} = f_{j,b}(i+1)$. After completing the OLE protocol with all parties $ P_j $, $ P_0 $ obtains:
\begin{equation*}
    z_0^v = (r_1^v + \cdots + r_{n-1}^v) \cdot e_k^0 + (a_1^{0,v} + \cdots + a_{n-1}^{0,v})
\end{equation*}
and $ P_i $ obtains:
\begin{equation*}
z_1^v = -(r_1^v + \cdots + r_{n-1}^v) \cdot e_k^i + (a_1^{1,v} + \cdots + a_{n-1}^{1,v}).
\end{equation*}
At this point, if $ P_0 $ and $ P_i $ hold the same element, then after completing the third phase, which involves the OPPRF, $ P_0 $ obtains:
\begin{equation*}
y_k^i = s_k^{0,i} + z_1^v = s_k^{0,i}  - (r_1^v + \cdots + r_{n-1}^v) \cdot e_k^i + (a_1^{1,v} + \cdots + a_{n-1}^{1,v}).
\end{equation*}
For each element $e_k^0 \in S_0$, there are $ \beta $ OLE results  related to each $P_i$ for $P_0$. Therefore, it is necessary for $P_0$ to use the OPPRF with each party $P_i$ to conditionally receive the OLE index $v$ corresponding to $ e_k^0$, so that it can retrieve the corresponding OLE output, denoted as $ z_0^v,v \in [\beta]$, and obtain the updated share from $P_i$.
After completing the above steps, we observe that if $ P_0 $ and $ P_i $ hold the same element, i.e., $ e_k^0 = e_k^i $, $ P_0 $ obtains the correctly updated share from party $P_i$. Otherwise, $P_0$ obtains a random value. The details are as follows:
\begin{align*}
y_k^i &= y_k^i + z_0^v \\
&= s_k^{0,i}  + (r_1^v + \cdots + r_{n-1}^v) \cdot e_k^0 + (a_1^{0,v} + \cdots + a_{n-1}^{0,v}) \\
&\quad - (r_1^v + \cdots + r_{n-1}^v) \cdot e_k^i + (a_1^{1,v} + \cdots + a_{n-1}^{1,v}) \\
&= s_k^{0,i}  + ((r_1^v + \cdots + r_{n-1}^v)\cdot(e_k^0 - e_k^i) \\
&\quad +((a_1^{0,v} + a_1^{1,v})+ \cdots +(a_{n-1}^{0,v} + a_{n-1}^{1,v})) \\
&= s_k^{0,i}  + (f_{1,b}(i+1) + \cdots + f_{n-1,b}(i+1))
\end{align*}

\textcolor{black}{
Ultimately, for each element $e_k^0 \in S_0$, $P_0$ performs secret reconstruction using the $n$ updated shares $y_k^i$. In each attempt, $P_0$ selects a subset of $t - 1$ shares from the other $n - 1$ parties, combines them with its own share, and applies Lagrange interpolation to compute $Recon(y_k^{i}) = f_k(\cdot)$. If any reconstruction satisfies $f_k(0) = e_k^0$, then $e_k^0$ is identified as an intersection element. Furthermore, $P_0$ identifies all holders of this element by checking whether $f_k(i+1) = y_k^i$ holds for each $i \in [1, n - 1]$.
}


\begin{figure*}
    \centering
    \makebox[\textwidth]{
    \fbox{
        \parbox{\dimexpr\textwidth-2\fboxsep-2\fboxrule}
        {
            \textsc{Parameters}: 
            \begin{itemize}
                \item The same as the protocol described in Fig. \ref{fig:ET-OT-MP-PSI Protocol}.
            \end{itemize}

            \textsc{Input}:
            \begin{itemize}
                \item Each party $P_i$ uses its private set $S_i$ as input.
                \item The threshold $t$. 
            \end{itemize}

            \textsc{Protocol}:
             \\  $\underline{\mathit{Conditional \ Secret \ Sharing}} $
            \begin{itemize}
                \item [(1)] Each party maps their set $S_i$ into bins using Cuckoo  and Simple hashing as described in Section \ref{sub:hashing schemes}, obtaining $B_S[\cdot]$ and $B_C[\cdot]$. \textcolor{black}{Each empty bin in $B_C[\cdot]$ is padded with a dummy element, while each bin in $B_S[\cdot]$ is padded with dummy elements to the maximum bin size $\beta$.}
                \item [(2)] For each element $e_k^0 \in S_0$, $P_0$ generates a corresponding random value $ e_k^0{'} $.
                Subsequently, $ P_0 $ uses $ e_k^0{'} $ as the secret and performs a $(t, n)$-Shamir's secret sharing, generating $ n $ shares $ s_k^{0,0}, \dots, s_k^{0,n-1} $.
                \item [(3)] For the $b^{\text{th}}$ bin, $b \in [m_b]$, $P_0$ and each party $P_i, i \in [1,n-1]$ execute an OPPRF.
                \begin{itemize}
                    \item $P_0$ is the sender with input $\{(e_k^0,s_k^{0,i})|e_k^0 \in B_S[b]\}$.
                    \item $P_i$ is the receiver with input $\{e_k^i|e_k^i \in B_C[b]\}$ and obtains a corresponding output $\hat{s}_k^{0,i}$ for every $e_k^i \in S_i$.
                \end{itemize}
            \end{itemize}
             $\underline{\mathit{Secret \ Shares \ Update}}$
            \begin{itemize}
                \item [(4)] For  the $b^{\text{th}}$ bin, $P_j$ performs secret shares update mentioned in Section \ref{sub:Secret shares updating}, generating $n$ shares $(i+1,f_{j,b}(i+1)),i,j \in [n]$ and directly sends the value $f_{j,b}(1)$ to $P_0$. For $e_k^0 \in B_C[b]$, $P_0$ updates its shares: $y_k^0 = s_k^{0,0}+f_{1,b}(1)+\cdots+f_{n-1,b}(1)$.
                \item [(5)] For each bin $B_C[b]$, $P_0$ and $P_j$ execute $\beta$ instances of the OLE protocol, where $\beta$ is the maximum bin size of $B_S$.
                \begin{itemize}
                    \item  $P_j$ acts as the sender with inputs $r_j^v$ and $a_j^{0,v}$, where $v \in [\beta]$ and $v$ is the index of OLE execution. Both inputs are random values. 
                    \item  $P_0$ acts as the receiver with input $e_k^0 \in B_C[b]$, which may be a dummy element. The OLE output for $P_0$ is $e_k^0 \cdot r_j^v + a_j^{0,v}$.
                \end{itemize}
                After completing the OLE protocol with all $P_j$, $P_0$ obtains $z_0^{v} = (r_1^v + \cdots + r_{n-1}^v) \cdot e_k^0 + (a_1^{0,v} + \cdots + a_{n-1}^{0,v})$.
                \item [(6)] For each bin $B_S[b]$, $P_i$ and $P_j$ execute $\beta$ instances of the OLE protocol.
                \begin{itemize}
                    \item $P_j$ acts as the sender with inputs $-r_j^v$ and $a_j^{1,v}$, where $a_j^{0,v} + a_j^{1,v} = f_{j,b}(i+1)$ and $v$ is the index of OLE execution.
                    \item $P_i$ acts as the receiver with input $e_k^i \in B_S[b]$, which may be a dummy element. The OLE output for $P_i$ is $e_k^i \cdot -r_j^v + a_j^{1,v}$.
                \end{itemize}
                After completing the OLE protocol with all $P_j$, $P_i$ obtains $z_1^v = -(r_1^v + \cdots + r_{n-1}^v) \cdot e_k^i + (a_1^{1,v} + \cdots + a_{n-1}^{1,v})$.
            \end{itemize}
             $\underline{\mathit{Conditional\ Collection\ and\ Reconstruction}}$
            \begin{itemize}
                \item [(7)] For  the $b^{\text{th}}$ bin, $P_0$ and $P_i$ invoke an OPPRF.
                \begin{itemize}
                    \item $P_i$ is the sender with input $\{(e_k^i,\mu_k^{0,i})|e_k^i \in B_S[b]\}$, where $\mu_k^{0,i} = \hat{s}_k^{0,i} + z_1^v $.
                    \item $P_0$ is the receiver with input $\{e_k^0|e_k^0 \in B_C[b]\}$ and obtains a corresponding output $y_k^i$ for each $e_k^0 \in S_0$. 
                \end{itemize}
                \item [(8)] For  the $b^{\text{th}}$ bin, $P_0$ and $P_i$ invoke an OPPRF.
                \begin{itemize}
                    \item $P_i$ is the sender with input $\{(e_k^i,v)|e_k^i \in B_S[b]\}$, where $v$ is the index of OLE execution about element $e_k^i$ in $B_S[b]$.
                    \item $P_0$ is the receiver with input $\{e_k^0|e_k^0 \in B_C[b]\}$ and obtains a corresponding output $v'$ for each $e_k^0 \in S_0$. 
                \end{itemize}
                \item [(9)] \textcolor{black}{For each $e_k^0$, $P_0$ identifies its position $b$ in $B_C$, retrieves $z_0^{v'}$ from the $b^{\text{th}}$ bin and updates the share as $y_k^i = y_k^i + z_0^{v'}$. 
                Then, for each element $e_k^0 \in S_0$, $P_0$ applies Lagrange interpolation to compute $Recon(y_k^i)=f_k(\cdot)$ over all subsets of $t$ shares among the $n$ values, always including its own share. If any reconstruction yields $f_k(0) = e_k^0{'}$, the element $e_k^0$ is confirmed to be in the intersection. Subsequently, $P_0$ determines all parties holding this element by verifying whether $f_k(i+1) = y_k^i$ for each $i$. Finally, $P_0$ obtains the complete intersection set $I$.
                }

                
            \end{itemize}         
        }
    }
    }
    \caption{ST-OT-MP-PSI protocol.}
    \label{fig:ST-OT-MP-PSI Protocol}
    \vspace{-0.5cm}
\end{figure*}

\section{Theoretical Analysis}
\textcolor{black}{
\subsection{Correctness and Security Analysis}\label{security analysis}
\begin{theorem}
\textit{The ET-OT-MP-PSI realizes the functionality $\mathcal{F}_{\text{T-OT-MP-PSI}}^{n,m,t}$ and is secure against collusion among up to $t - 2$ parties in the semi-honest model, given the statistical security parameter $\lambda$ and the computational security parameter $\kappa$.}\label{theorem 1}
\end{theorem}
\begin{proof}
    The proof consists of two parts: correctness and security. Due to the page limit, we put the formal correctness and security proofs of the ET-OT-MP-PSI in Appendix \ref{PROOF OF ET-OT-MP-PSI}
\end{proof}
}
\textcolor{black}{
\begin{theorem}
\textit{The ST-OT-MP-PSI realizes the functionality $\mathcal{F}_{\text{T-OT-MP-PSI}}^{n,m,t}$ and is secure against collusion among up to $n-1$ parties in the semi-honest model, given the statistical security parameter $\lambda$ and the computational security parameter $\kappa$. } 
\end{theorem}
\begin{proof}
    Similar to Theorem \ref{theorem 1}, the proof also consists of correctness and security. For brevity, we defer the full proof of the ST-OT-MP-PSI protocol in Appendix \ref{PROOF OF ST-OT-MP-PSI}
\end{proof}
}

\subsection{Complexity Analysis} \label{sub:Complexity Analysis}
The complexity comparison between our Traceable OT-MP-PSI protocols and related work with traceability \cite{mahdavi2020practical} is illustrated in Table \ref{tab:protocols_complexity}. 
To ensure consistency in comparison, we refer to $ P_0 $ as the Leader and all other participants $ P_i $ as the Clients throughout this section.

\begin{table*}
\caption{Analytic comparison of related work with our protocols. $n$ is the number of parties. $t$ is the threshold. Each party holds a set of size $m$. $\lambda$ and $\kappa$ are statistical and computational security parameters, respectively. $e$ is Euler’s constant.}
\centering
\setlength{\tabcolsep}{13pt}
\renewcommand{\arraystretch}{1.3}
\begin{tabular}{c|c|c|c|c|c}
\toprule
\multirow{2}{*}{\textbf{Protocol}} & \multicolumn{2}{c|}{\textbf{Communication}} & \multicolumn{2}{c|}{\textbf{Computation}} & \textbf{Corruption} \\ 
\cline{2-5}
 & \textbf{Leader} & \textbf{Client} & \textbf{Leader} & \textbf{Client} & \textbf{Resilience} \\ \hline
\hline
Mahdavi et al. \cite{mahdavi2020practical} & \multicolumn{2}{c|}{$O(nmt)$}
 & \multicolumn{2}{c|}{$O(m(nlog(\frac{m}{t}))^{2t})$} & $t-2^{*}$ \\ 
\hline
ET-OT-MP-PSI & $O(nm\lambda)$ & $O(nm\lambda)$ & $O(max\{t^2(\frac{e(n-1)}{t-1})^{t-1}),n\kappa\}m)$ & $O(max\{\kappa,nt\lambda\}m)$ & $t-2$ \\ 
\hline
ST-OT-MP-PSI & $O(n^2m\lambda)$ & $O(nm\lambda )$ & $O(max\{t^2(\frac{e(n-1)}{t-1})^{t-1}),n^2\lambda, n\kappa\}m)$ & $O(max\{\kappa,nt\lambda\}m)$ & $n-1$ \\ 
\bottomrule
\end{tabular}

\begin{flushleft}
\footnotesize
$^{*}$ In Mahdavi et al.'s protocol, it is required that certain designated roles must not collude. 
\end{flushleft}

\label{tab:protocols_complexity}
\end{table*}

\noindent\textbf{Communication complexity}.  
In ET-OT-MP-PSI, the \textcolor{black}{\textit{conditional secret sharing} phase} involves the Leader engaging with $n-1$ Clients to execute the OPPRF protocol, incurring a communication complexity of $O(nm\lambda)$. 
\textcolor{black}{As part of the \textit{secret shares update} step,} each Client transmits $m_b$ values to every other participant to update shares, resulting in an additional communication cost of $O(nm\lambda)$. 
The subsequent OPPRF with the Leader \textcolor{black}{in the \textit{conditional collection and reconstruction} phase} contributes a further $O(m\lambda)$ to the overall communication complexity.
 In ST-OT-MP-PSI, \textcolor{black}{the \textit{secret share update} procedure} requires the Leader to perform $n^2m_b$ additional OLE protocols, which involve the exchange of $O(n^2m\lambda)$ ciphertexts. Similarly, each Client executes $nm_b$ OLE protocols, corresponding to $O(nm\lambda)$ ciphertexts.

 \noindent\textbf{Computation complexity}. 
In ET-OT-MP-PSI, \textcolor{black}{during \textit{conditional secret sharing}}, the Leader generates $m$ polynomials of degree $t-1$ and evaluates them at $n$ points, resulting in a computational complexity of $O(nmt)$. The Leader also performs OPPRF with $n-1$ Clients, adding $O(nm\kappa)$. 
\textcolor{black}{For \textit{reconstruction}}, since the Leader’s shares are always correct, each of the $m$ elements requires $\binom{n-1}{t-1}$ operations, leading to a total complexity of $O(mt^2\binom{n-1}{t-1})$.
Using the approximation $\binom{n-1}{t-1} \leq (e(n-1)/(t-1))^{t-1}$, the complexity is relaxed to $O(mt^2(e(n-1)/(t-1))^{t-1})$, where $e$ is Euler’s constant. 
\textcolor{black}{In the phase of \textit{secret shares update}}, each Client generates $m_b$ polynomials of degree $t-1$ and evaluates them at $n$ points, contributing $O(nmt\lambda)$. Additionally, \textcolor{black}{for \textit{conditional collection}}, OPPRF between each Client and the Leader adds $O(m\kappa)$.
In ST-OT-MP-PSI, \textcolor{black}{during \textit{secret shares update}}, the Leader performs $n^2m_b$ additional OLE protocols, requiring $O(n^2m\lambda)$ encryptions and decryptions in our implementation. Clients execute $nm_b$ OLE protocols, involving $O(nm\lambda)$ homomorphic operations.

\noindent\textbf{\textcolor{black}{Comparision to Mahdavi et al.'s. protocol [17].} }
\textcolor{black}{
For communication complexity, our protocols maintain better scalability by avoiding any dependence on the threshold $t$. In contrast, Mahdavi et al.'s protocol incurs communication costs that grow with $t$, limiting its practicality in high-threshold settings.
For the computational complexity, compared to the computational complexity $O(m(n log(m/t))^{2t})$ of Mahdavi et al.’s protocol [17], our Traceable OT-MP-PSI protocols demonstrate better efficiency. Specifically, although both our protocols and Mahdavi et al.’s protocol similarly have exponential complexity with $t$, our protocols achieve a much smaller base and exponent $t$ rather than $2t$, resulting in less computational cost. 
The enhancement is primarily attributed to the integration of OPPRF with Shamir's secret sharing, which allows each intersection element to be precisely associated with its corresponding shares, instead of exhaustively trying all possible shares from the involved parties in Mahdavi et al.’s protocol when performing reconstruction, thereby significantly reducing the reconstruction time.
Lastly, for security, compared to the $ t-2 $ corruption tolerance in Mahdavi et al.'s protocol, our ET-OT-MP-PSI protocol achieves the same level of resistance while eliminating the assumption that certain special parties do not collude. Meanwhile, our ST-OT-MP-PSI protocol further strengthens security by tolerating collusion among up to $ n-1 $ semi-honest parties.
}

\section{Performance Evaluation}

\subsection{Implementation and Experimental Settings}
To evaluate the performance of the proposed Traceable OT-MP-PSI protocols, we implemented both protocols in C++\footnote{\url{https://github.com/Yank3l/T-OT-MP-PSI}}. The implementation relies on the NTL library\footnote{\url{https://libntl.org/}} for large number operations. Communication between the parties is handled using the Boost library, which provides robust tools for message-passing, networking, and parallel processing. 

\textcolor{black}{We implemented Shamir’s secret sharing using the NTL library, instantiating the finite field modulo the largest 128-bit prime \(p\), which allows us to accommodate the widest possible range of 128-bit elements and aligns with real-world deployment requirements. 
We adopt the table-based OPPRF construction of Kolesnikov et al.~\cite{kolesnikov2017practical}, which has favorable communication and computational cost. 
To satisfy the security assumption that the OPPRF and the Shamir's secret sharing operate over the same finite field $\mathbb{F}_p$, in the implementation, we adjust the table-based OPPRF by replacing XOR operations with modular addition in step 3, and modular subtraction in step 6, respectively. Note that, the programmed and the non-programmed points share the same distribution with such adjustment, and hence the receiver cannot distinguish between the programmed and non-programmed entries.
}

In the ST-OT-MP-PSI, we utilize the OLE protocol proposed by de Castro et al. \cite{de2021fast}\footnote{\url{https://github.com/leodec/ole_wahc}}, which is based on Ring Learning with Errors (RLWE).
However, the chosen OLE code does not natively support a 128-bit plaintext modulus. To address this limitation, we follow the method described in Section 5.2 of their paper and select $ p $ as the product of four smaller 32-bit prime numbers, i.e., $p = \prod_{i=0}^{3} p_i$, thereby extending the original OLE to support a 128-bit plaintext modulus. Since the modulus $p$ is the product of prime numbers, the implementation of this protocol leverages the Chinese Remainder Theorem (CRT). \textcolor{black}{Nevertheless, because we decompose each 128-bit share into four 32-bit CRT residues, our ST-OT-MP-PSI instantiation performs four independent OPPRF evaluations for every share distribution and reconstruction. To this end, we select the four largest 32-bit primes as CRT moduli. Although reducing each residue modulo a 32-bit prime introduces a larger bias than using a single 128-bit modulus, the joint distribution of programmed and non-programmed points remains identical and thus is computationally indistinguishable to the adversary.}

Our benchmarking experiments were conducted on a cloud server equipped with an Intel(R) Xeon(R) CPU running at 3.1GHz, featuring 80 vCores and 192GB of RAM, and operating on Ubuntu 22.04. In our experimental setup, each participant operated within a single process, and communication was conducted over a local network without bandwidth or latency constraints. The length of each element is 128 bits.
To better evaluate the performance of the proposed protocols, we performed experiments under varying settings of participant numbers and set sizes. 

We pick Mahdavi et al.'s protocol \cite{mahdavi2020practical} as a comparison baseline since both protocols similarly provide traceability in OT-MP-PSI. This comparison highlights that our protocols achieve significantly higher efficiency compared to the existing solution, while maintaining the same functionality.
The publicly available source code enables direct implementation\footnote{\url{https://github.com/cryspuwaterloo/OT-MP-PSI}} and consistent benchmarking under similar conditions.
To ensure a fair comparison, we adopted the same elements generation method as that used by Mahdavi et al. Among the two constructions presented in their work, we concentrated on the more efficient variant for benchmarking.

\subsection{Results Evaluation}

Tables \ref{protocol 1 result} and \ref{protocol 2 result} present the performance of our proposed Traceable OT-MP-PSI protocols for varying numbers of participants $ n $ and the set sizes $ m $. The results indicate a clear linear relationship between runtime of protocols and set size.
This observation is consistent with the computational complexity analysis in Section \ref{sub:Complexity Analysis}, where the complexity scales linearly with the set size $m$.
For example, in Table \ref{protocol 1 result}, with $ n = 5 $ participants and a threshold of $ t = 3 $, the runtime increases from 1.73s at $ m = 2^{14} $ to 6.23s at $ m = 2^{16} $ and 24.76s at $ m = 2^{18} $. 

The ET-OT-MP-PSI demonstrates strong performance, achieving a runtime of 45.21s for $ n = 10 $, $ t = 5 $, and $ m = 2^{16} $. Meanwhile, the ST-OT-MP-PSI achieves enhanced security by introducing the OLE, though at the cost of increased computational overhead. 
Specifically, this protocol requires an additional $O(n^2m\lambda)$ executions of the OLE protocol, which imposes additional performance overhead. 
For instance, with $ n = 5 $, $ t = 3 $, and $ m = 2^{16} $, the protocol completes in approximately 207s.
In practice, the choice between the two protocols depends on the specific balance between performance and security requirements. The ET-OT-MP-PSI is ideal for scenarios prioritizing speed, while the ST-OT-MP-PSI is better suited for scenarios where robust security is essential.

\begin{table}[h]
        \caption{\textcolor{black}{The average runtime (in seconds) over 10 trials of the ET-OT-MP-PSI.}}
        \label{protocol 1 result}
        \setlength{\tabcolsep}{8pt}
        \renewcommand{\arraystretch}{1.5}
        \begin{tabular}{c|c|c|c|c|c|c}
            \toprule
             \multirow{2}{*}{$\boldsymbol{m}$} & \multicolumn{6}{c}{$\boldsymbol{(n,t)}$} \\ \cline{2-7} 
              & (5,3) & (6,3) & (7,4) & (8,4) & (9,5) & (10,5) \\ \hline \hline
            $2^{12}$        & \textcolor{black}{0.68}  & \textcolor{black}{0.87}  & \textcolor{black}{1.23}  & \textcolor{black}{1.50}  & \textcolor{black}{2.62}  & \textcolor{black}{3.53}   \\ \hline
            $2^{14}$     & \textcolor{black}{1.73}  & \textcolor{black}{2.15}  & \textcolor{black}{3.34}  & \textcolor{black}{4.34}  & \textcolor{black}{8.37}  & \textcolor{black}{11.88}  \\ \hline
            $2^{16}$         & \textcolor{black}{6.23}  & \textcolor{black}{7.67}  & \textcolor{black}{12.32} & \textcolor{black}{15.71} & \textcolor{black}{31.80} & \textcolor{black}{45.21}  \\ \hline
            $2^{18}$         & \textcolor{black}{24.76} & \textcolor{black}{30.43} & \textcolor{black}{48.66} & \textcolor{black}{61.50} & \textcolor{black}{128.85} & \textcolor{black}{182.25} \\  
            \bottomrule
        \end{tabular}
        \vspace{-0.5cm}
\end{table}

\begin{table}[h]
        \caption{\textcolor{black}{The average runtime (in seconds) over 10 trials of the ST-OT-MP-PSI.}}
        \label{protocol 2 result}
        \setlength{\tabcolsep}{7pt}
        \renewcommand{\arraystretch}{1.5}
        \begin{tabular}{c|c|c|c|c|c|c}
        \toprule
        \multirow{2}{*}{$\boldsymbol{m}$} & \multicolumn{6}{c}{$\boldsymbol{(n,t)}$} \\
        \cline{2-7}
                         & (5,3) & (6,3) & (7,4) & (8,4) & (9,5) & (10,5) \\ 
        \hline \hline
         $2^{12}$                & \textcolor{black}{14.67}  & \textcolor{black}{20.55}  & \textcolor{black}{28.11}  & \textcolor{black}{36.31}  & \textcolor{black}{47.51}  & \textcolor{black}{60.04}   \\ \hline
        $2^{14}$                & \textcolor{black}{53.22}  & \textcolor{black}{76.22}  & \textcolor{black}{104.71}  & \textcolor{black}{135.91}  & \textcolor{black}{181.41}  & \textcolor{black}{229.09}  \\ \hline
        $2^{16}$                & \textcolor{black}{207.78}  & \textcolor{black}{298.64}  & \textcolor{black}{404.89} & \textcolor{black}{529.12} & \textcolor{black}{716.11} & \textcolor{black}{903.28}  \\ 
        \bottomrule
    \end{tabular}
          
    \end{table}

\begin{figure*}[ht]
\centering
\begin{subfigure}{0.49\linewidth}
 \centering
 \includegraphics[width=0.9\linewidth]{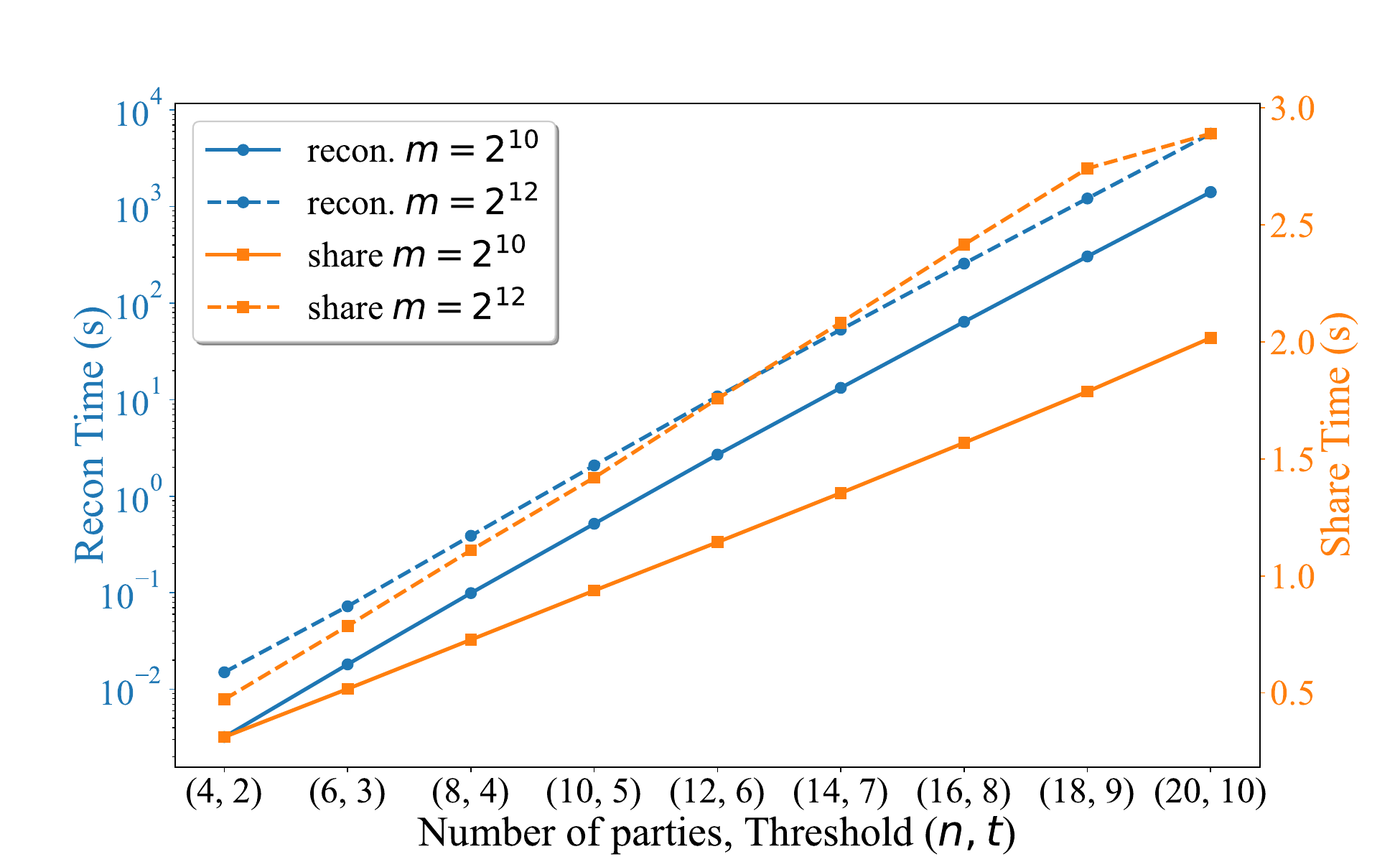}
 \caption{ET-OT-MP-PSI}
 \label{fig:sub1-ET-OT-MP-PSI}
\end{subfigure}
\hfill
\begin{subfigure}{0.49\linewidth}
\centering
\includegraphics[width=0.9\linewidth]{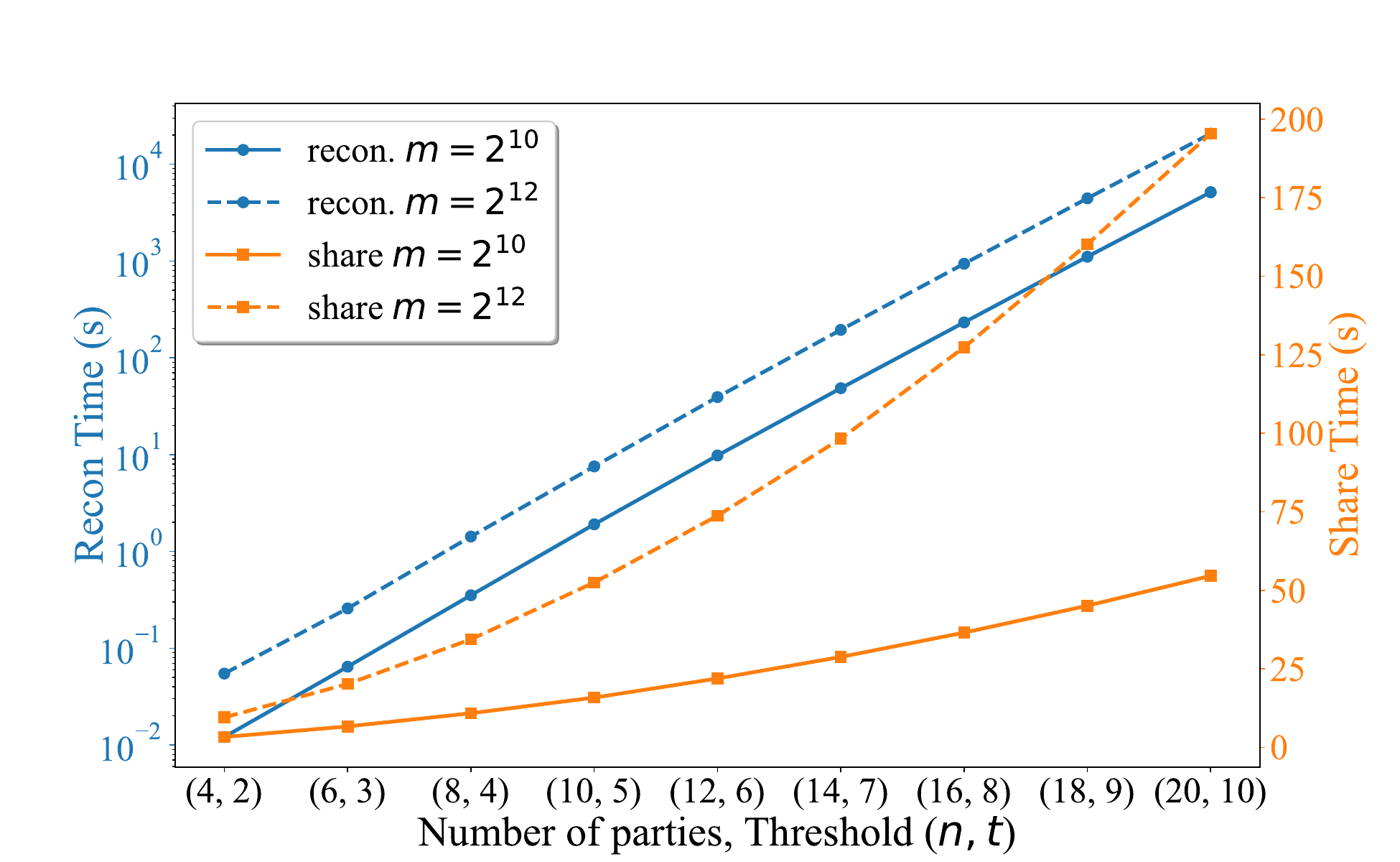}
\caption{ST-OT-MP-PSI}
\label{fig:sub2-ST-OT-MP-PSI}
\end{subfigure}
\caption{\textcolor{black}{The average runtime (in seconds) of our protocols with varying numbers of participants.}}
\label{fig:Performance of our protocols with varying numbers of parties.}
\vspace{-0.5cm}
\end{figure*}

\textcolor{black}{
Fig.~\ref{fig:Performance of our protocols with varying numbers of parties.} illustrates the runtime performance of our protocol under varying numbers of parties $n$ with a threshold setting of $t = n/2$, evaluated for two different set sizes: $m = 2^{10}$ and $m = 2^{12}$. We separately measure the runtime of the two phases of the protocol: the share phase, which includes both the initial Shamir's secret sharing and subsequent share updating, and the reconstruction phase, which performs secret reconstruction.
It is worth noting that the left $y$-axis is presented on a logarithmic scale, while the right $y$-axis uses a normal (linear) scale to better illustrate the growth trends.
When the threshold is fixed at $t = n/2$, we observe that in the ET-OT-MP-PSI, the share phase scales approximately linearly with the number of parties, whereas in the ST-OT-MP-PSI, it exhibits quadratic growth with respect to the number of parties. In contrast, the reconstruction phase shows clear exponential growth with $n$, which is expected since reconstructing each secret requires iterating over all possible subsets of $t$ shares.
}


\begin{figure*}[!h]
\centering
\begin{subfigure}{0.32\linewidth}
    \centering
    \includegraphics[width=\linewidth]{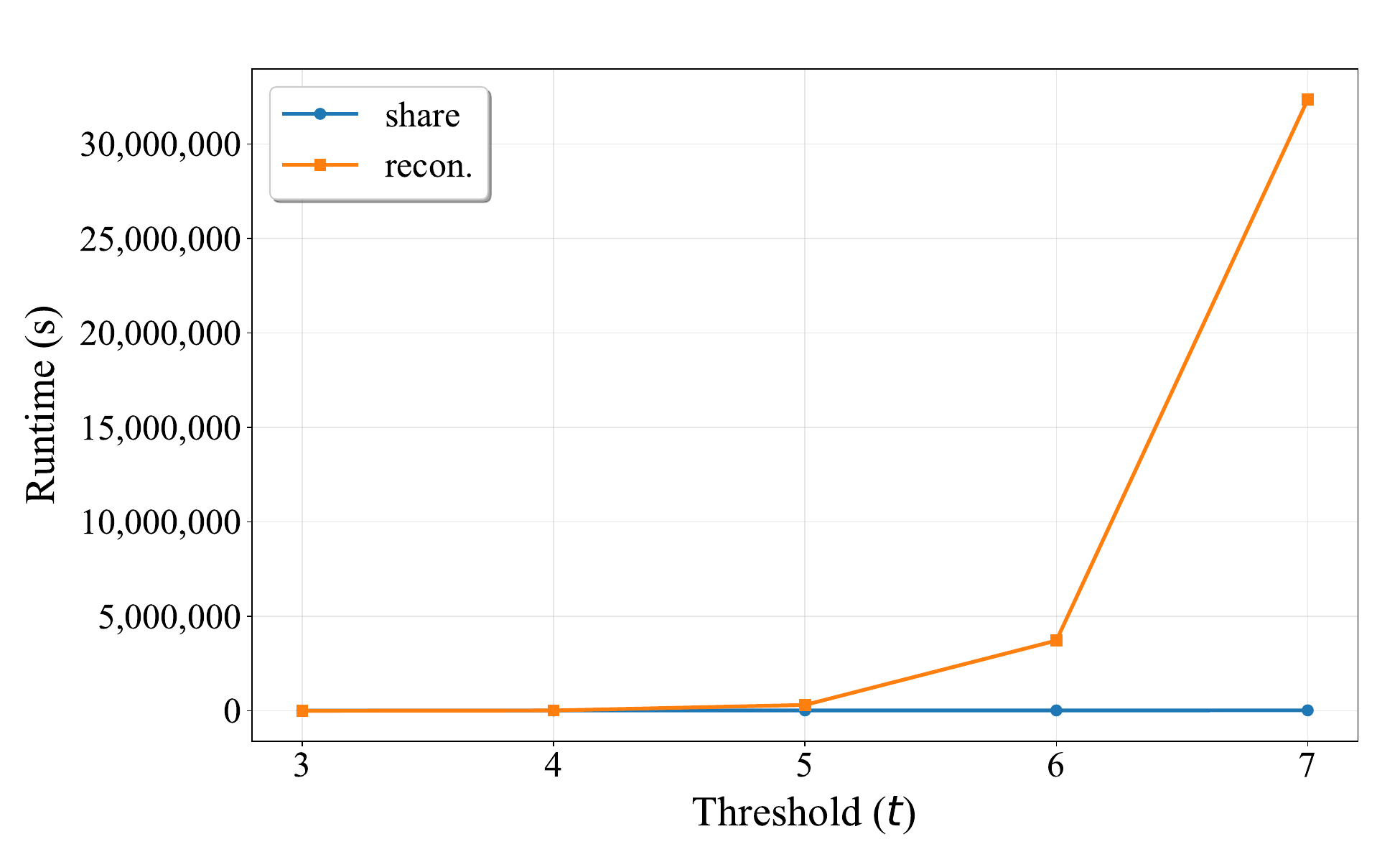}
    \caption{Mahdavi et al.'s protocol}
    \label{fig:sub1}
\end{subfigure}
\hfill
\begin{subfigure}{0.32\linewidth}
    \centering
    \includegraphics[width=\linewidth]{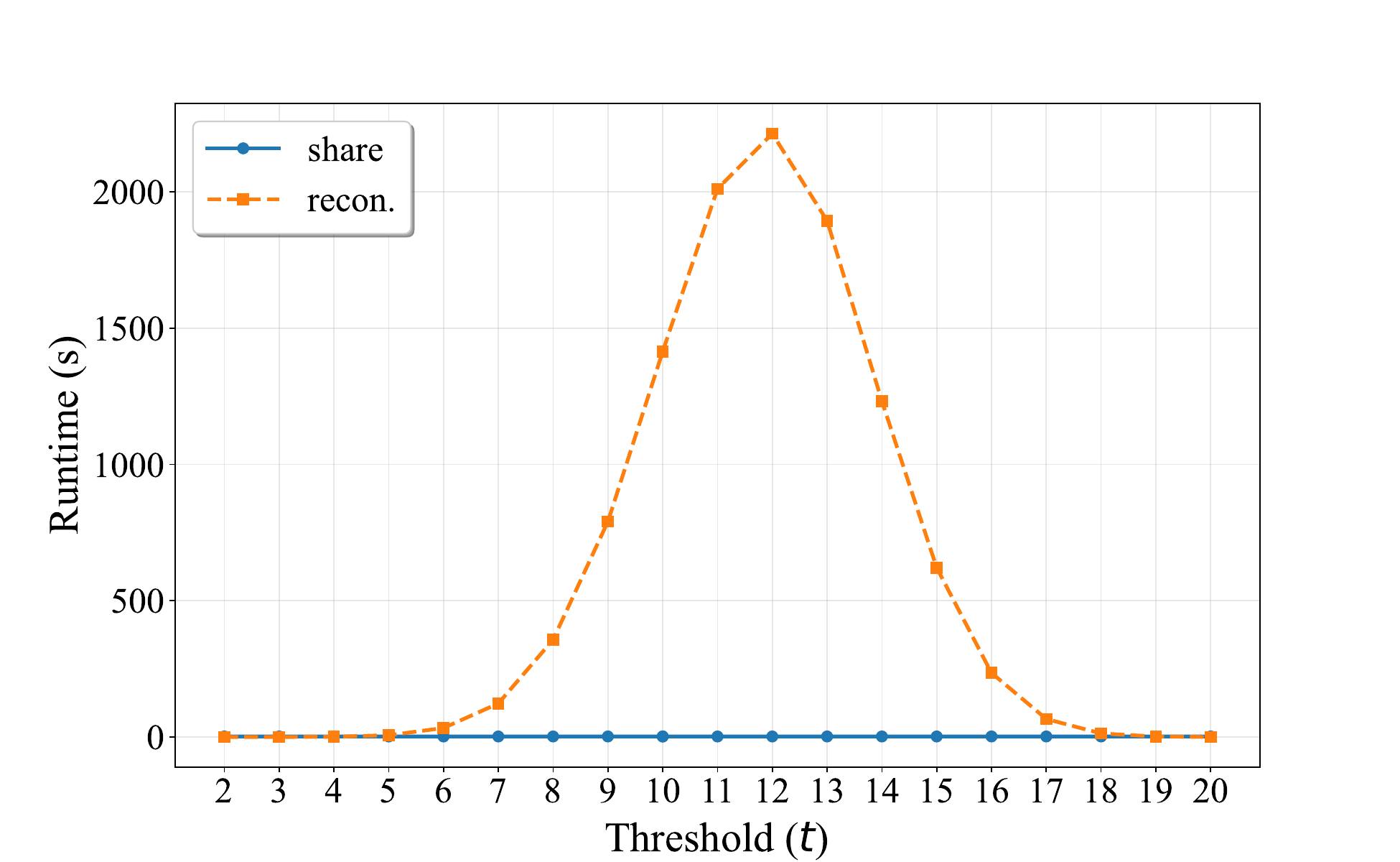}
    \caption{ET-OT-MP-PSI}
    \label{fig:ET-OT-MP-PSI}
\end{subfigure}
\hfill
\begin{subfigure}{0.32\linewidth}
    \centering
    \includegraphics[width=\linewidth]{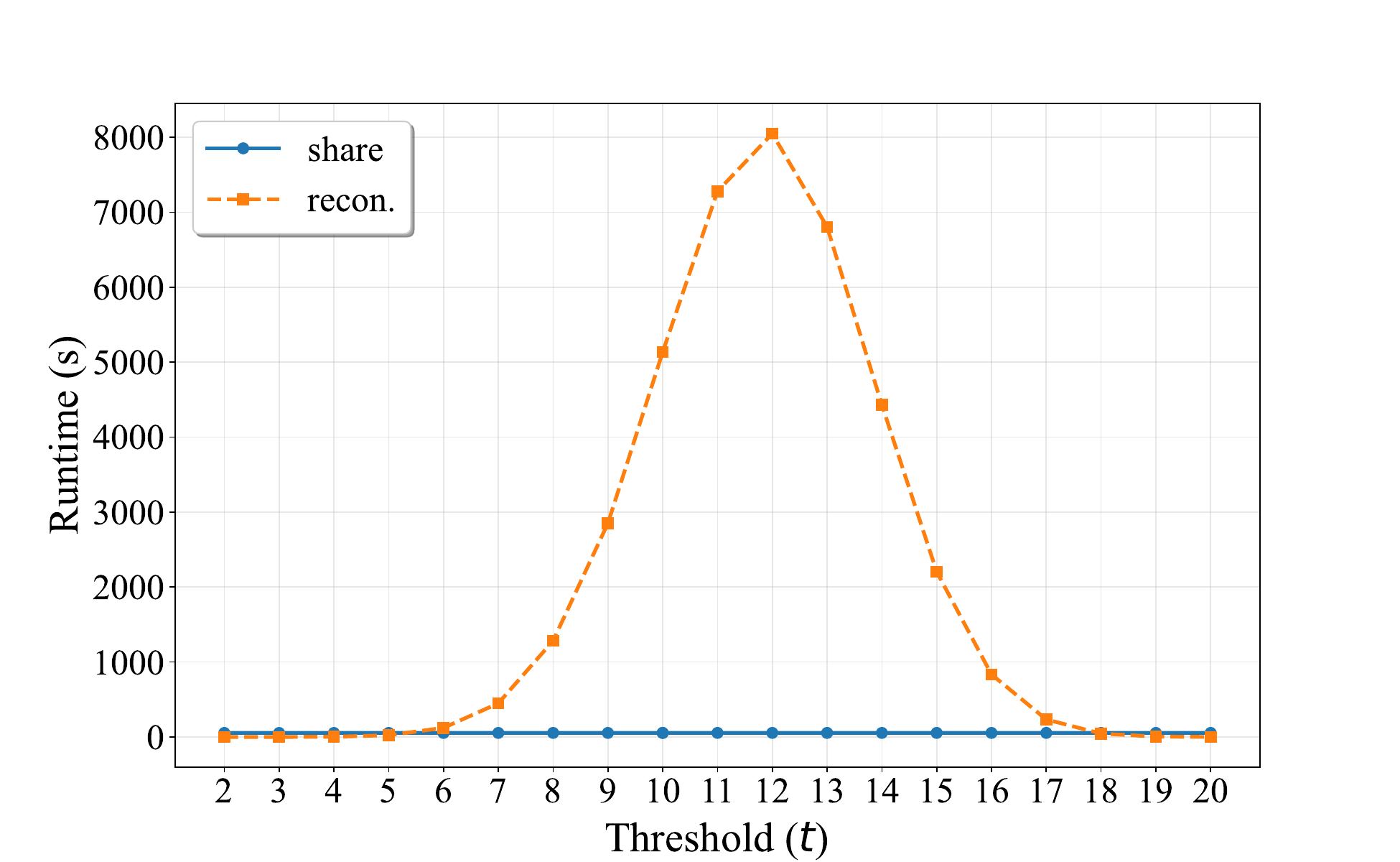}
    \caption{ST-OT-MP-PSI}
    \label{fig:ST-OT-MP-PSI}
\end{subfigure}
\caption{\textcolor{black}{Comparison of the runtime (in seconds) between Mahdavi et al.'s protocol and our protocols for varying threshold $t$.}}
\label{fig:Running time (in seconds) for varying threshold $t$}
\vspace{-0.5cm}
\end{figure*}

\textcolor{black}{
Fig.~\ref{fig:Running time (in seconds) for varying threshold $t$} shows the performance of our protocols and  Mahdavi et al.'s protocol \cite{mahdavi2020practical} across different threshold values $t$.  Due to the excessive runtime of Mahdavi et al.'s protocol, we evaluated it under a smaller setting with $n = 10$ and $m = 2^5$, whereas our protocols were tested up to larger parameters with $n = 20$ and $m = 2^{10}$. Notably, when $t > 7$, Mahdavi et al.'s protocol exceeds the evaluation time limit and is thus not presented in the figure. 
In the reconstruction phase, the runtime of both our protocols and Mahdavi et al.'s grows exponentially with the threshold $t$, which aligns with our computational complexity analysis.
We observe that Mahdavi et al.'s protocol exhibits a rapid increase in reconstruction runtime, which continues to grow until $t = n$, making it suitable only for small values of $t$. In contrast, for our protocols, the reconstruction runtime reaches its peak when the number of participants $n$ and set size $m$ are fixed, and the threshold $t$ approaches $(n+1)/2$. This behavior is expected, as the number of combinations $\binom{n-1}{t-1}$ is maximized near this point. Due to the smaller base and exponent, our protocols experience a much slower rate of growth, and during the growth phase, the runtime is consistently much smaller than that of Mahdavi et al.'s protocol for the same settings.
}

In the performance comparison, we focused on the more efficient version of the protocol proposed by Mahdavi et al. \cite{mahdavi2020practical}. To provide a comprehensive evaluation, we tested the protocols in two distinct scenarios: one involving a larger number of participants with smaller sets and the other featuring fewer participants with larger sets.
For the first scenario, with more participants and smaller sets, we set the number of participants to 10, the threshold to 5, and the set sizes to $ 2^4, 2^5, 2^6, $ and $ 2^7 $. In the second scenario, with fewer participants and larger sets, we set the number of participants to 5, the threshold to 3, and the set sizes to $ 2^{10}, 2^{12}, 2^{14}, $ and $ 2^{16} $. Table \ref{comparison} summarizes the performance of the protocols under these settings. 

The results demonstrate that both the ET-OT-MP-PSI and  ST-OT-MP-PSI consistently surpass Mahdavi et al.'s protocol in shares generation, reconstruction, and overall runtime across all evaluated scenarios.
For instance, in terms of overall runtime, with 10 participants, a threshold of 5, and a set size of $ 2^7 $, our protocols are \textcolor{black}{ 4312$ \times $} and \textcolor{black}{ 637$\times $} faster, respectively, compared to Mahdavi et  al.'s protocol. Similarly, with 5 participants, a threshold of 3, and a set size of $ 2^{14} $,
\textcolor{black}{
our protocols achieve speedups of  15056$ \times $ and  505$ \times $, respectively.
To gain deeper insights into the performance advantages of our protocols, we analyze the share and reconstruction phases individually.
In the share phase, ET-OT-MP-PSI utilizes a combination of Shamir's secret sharing and OPPRF. These techniques are predominantly based on efficient symmetric-key operations, which are computationally lightweight. In contrast, the share stage in Mahdavi et al.'s protocol relies on Paillier homomorphic encryption, which is significantly more computationally expensive. The ST-OT-MP-PSI further introduces OLE to enable secure share updates, which incurs a moderate computational overhead but remains more efficient than the homomorphic encryption used in Mahdavi et al.'s scheme. 
In the reconstruction phase, our protocols also demonstrate superior efficiency. For instance, with 10 participants, a threshold of 5, and a set size of $2^7$, our protocols are 45215$\times$ and 13761$\times$ faster, respectively, compared to Mahdavi et al.'s protocol.
By leveraging OPPRF and Shamir's secret sharing, we reduce the reconstruction complexity and significantly lowering the computational overhead required for reconstruction operations.}

These results clearly show the superior efficiency of our proposed protocols compared to Mahdavi et al.'s protocol. The consistent performance improvements across various settings highlight  the practicality of both the ET-OT-MP-PSI and ST-OT-MP-PSI protocols. By offering a flexible trade-off between runtime efficiency and security, our protocols are positioned to address a variety of real-world needs.

 \begin{table*}[!h]
    \caption{\textcolor{black}{Comparison of the overall runtime (in seconds) between Mahdavi et al.'s protocol and our proposed protocols across various settings.}}
    \label{comparison}
    \centering
    \setlength{\tabcolsep}{11pt}
    \renewcommand{\arraystretch}{1.3}
    \begin{tabular}{c|c|cccc|cccc}
    
        \toprule
        \multicolumn{2}{c|}{$\boldsymbol{(n,t)}$}   &      \multicolumn{4}{c|}{(10,5)} &   \multicolumn{4}{c}{(5,3)}  \\ 
        \hline
         \multicolumn{2}{c|}{$\boldsymbol{m}$}  & \centering $2^4$ & $2^5$ & $2^6$ & $2^7$ & $2^{10}$ & $2^{12}$ & $2^{14}$ & $2^{16}$  \\

        \hline \hline
        \multirow{3}{*}{\textbf{Mahdavi et al.} \cite{mahdavi2020practical}}        & share            & 83.24  & 167.05 & 335.73 & 672.90 & 1653.78 & 6585.70  & 26395.50 & -$^{*}$   \\
                                        & recon.            & 79.34  & 308.37  & 993.63  & 3165.06 & 13.31  & 62.99 & 404.94 & -  \\ 
                                        & total            & 162.58 & 475.42 & 1329.36 &3837.96 & 1667.09 & 6648.69  & 26800.44 & - \\
        \hline
        \multirow{3}{*}{\textbf{ET-OT-MP-PSI}}     & share$^{**}$     & \textcolor{black}{0.77}        & \textcolor{black}{0.77}         & \textcolor{black}{0.79}        & \textcolor{black}{0.79}   & \textcolor{black}{0.42}  & \textcolor{black}{0.65}  & \textcolor{black}{1.56} &\textcolor{black}{5.42}  \\
                                        & recon.           & \textcolor{black}{0.01}       & \textcolor{black}{0.02}       & \textcolor{black}{0.03}        & \textcolor{black}{0.07}  & \textcolor{black}{0.01}  & \textcolor{black}{0.05}  & \textcolor{black}{0.22} &\textcolor{black}{0.88}  \\ 
                                        & total           & \textcolor{black}{0.78}       & \textcolor{black}{0.79}        & \textcolor{black}{0.82}      & \textcolor{black}{0.86}  & \textcolor{black}{0.43}  &  \textcolor{black}{0.70} & \textcolor{black}{1.78} &\textcolor{black}{6.29}  \\
        \hline
        \multirow{3}{*}{\textbf{ST-OT-MP-PSI}}     & share$^{**}$      & \textcolor{black}{4.23}       & \textcolor{black}{4.60}        &  \textcolor{black}{4.99}       & \textcolor{black}{5.79}  & \textcolor{black}{4.84}  & \textcolor{black}{14.38}  & \textcolor{black}{52.26} &\textcolor{black}{205.04}  \\
                                        & recon.           & \textcolor{black}{0.03}       & \textcolor{black}{0.05}        &  \textcolor{black}{0.11}       & \textcolor{black}{0.23} & \textcolor{black}{0.05} & \textcolor{black}{0.19}  & \textcolor{black}{0.77}  & \textcolor{black}{3.06}   \\ 
                                        & total           &  \textcolor{black}{4.26}      & \textcolor{black}{4.66}        &  \textcolor{black}{5.10}       & \textcolor{black}{6.02} & \textcolor{black}{4.89}  &  \textcolor{black}{14.57} & \textcolor{black}{53.02}  & \textcolor{black}{208.11}  \\
        \bottomrule
    \end{tabular}
    \begin{flushleft}
    \footnotesize
    $^{*}$ Cells with ``-'' denote the task could not be completed within the testing time.  
    \\
    $^{**}$ ``share'' includes both secret sharing and shares update phase. 
    \end{flushleft}
    \vspace{-0.5cm}
\end{table*}

\section{Related Work}
\subsection{Multi-Party PSI and Variants }
With the wide range of applications for  MP-PSI, the past decade has witnessed the development of numerous protocols aimed at tackling challenges in both efficiency and security with diverse cryptographic techniques.
Freedman et al. \cite{freedman2004efficient} proposed the first MP-PSI in the semi-honest model, relying on oblivious polynomial evaluation (OPE) with homomorphic encryption. This approach was later adopted by other works, including \cite{cheon2012multi,dachman2011secure,hazay2017scalable,sang2007privacy}.
Miyaji and Nishida \cite{miyaji2015scalable} combined exponential ElGamal encryption with Bloom filters to design an MP-PSI, which relies on a trusted third party and is applied to medical data analysis  \cite{miyaji2017privacy}.
Kolesnikov et al. \cite{kolesnikov2017practical} proposed three constructions for instantiating oblivious programmable pseudorandom function (OPPRF) using oblivious transfer. Building on this, the authors combined zero-value secret sharing to develop the first efficient MP-PSI.
Inbar et al. \cite{inbar2018efficient} extended the two-party PSI by Dong et al. \cite{dong2013private} to a multi-party setting by utilizing the mergeability of garbled Bloom filters.
The first practically efficient MP-PSI with malicious security, introduced by Ben Efraim et al. \cite{ben2022psimple}, skillfully integrates techniques from semi-honest MP-PSI \cite{inbar2018efficient} and malicious two-party PSI \cite{rindal2017improved}, leveraging oblivious transfer and garbled Bloom filters.
Chandran et al. \cite{chandran2021efficient} proposed a modification to Kolesnikov et al.'s protocol \cite{kolesnikov2017practical}, replacing the construction of secret sharing that XOR to zero with Shamir’s secret sharing scheme, resulting in a more efficient MP-PSI.
In contrast to \cite{kolesnikov2017practical} and \cite{chandran2021efficient}, which rely on the OPPRF, Wu et al. \cite{wu2024ring} adopt the more efficient oblivious PRF (OPRF) and a data structure called the oblivious key-value store (OKVS), leading to the development of two MP-PSI: O-Ring and K-Star, designed to address distinct application needs.

Over time, MP-PSI has evolved into a broader family of protocols tailored to meet diverse application-specific needs. Notable variants of MP-PSI include multi-party private set intersection cardinality (MP-PSI-CA) \cite{debnath2021secure,trieu2022multiparty,liu2020quantum,shi2022quantum}, which calculates the size of the intersection without revealing the intersecting elements, and multi-party delegated PSI \cite{abadi2022multi}, which allows parties to outsource the storage of their datasets to a cloud computing service. 
Other extensions, such as MP-PSI-CA-sum \cite{arbitrary2023practical}, differ from MP-PSI-CA in that MP-PSI-CA-sum additionally outputs the sum of the associated integer values of all the data belonging to the intersection, providing richer insights beyond just the cardinality.

\subsection{Multi-Party PSI with Threshold}
As an extension of MP-PSI, variants with threshold such as T-MP-PSI and OT-MP-PSI have attracted the attention of researchers.
Kissner and Song \cite{kissner2004private} presented the first T-MP-PSI and OT-MP-PSI. They leveraged Paillier encryption to develop algorithms for encrypted polynomial operations. 
Miyaji and Nishida \cite{miyaji2015scalable} introduced a T-MP-PSI called d-and-over MPSI, combining  Bloom filters and exponential ElGamal encryption. However, the protocol relies on the assumption of a trusted third party. 
Mahdavi et al. \cite{mahdavi2020practical} proposed a new primitive called oblivious pseudo-random secret sharing (OPR-SS), which leverages oblivious pseudo-random function (OPRF) and Shamir's secret sharing. Building on this primitive, a new OT-MP-PSI with  traceability  was developed. To improve efficiency, the Paillier cryptosystem was introduced to reduce reconstruction time. Nevertheless, the optimized OT-MP-PSI  remains impractical for real-world use. Notably, Mahdavi et al.’s OT-MP-PSI  not only computes the intersection but also reveals the holders of the intersecting elements.
Bay et al. \cite{bay2021practical} presented a novel T-MP-PSI  that utilizes Bloom filters and threshold Paillier encryption. The protocol verifies whether an element is held by at least $ t $ participants through two rounds of multi-party secure comparison protocol (SCP). 
Chandran et al. \cite{chandran2021efficient} introduced a T-MP-PSI  called Quorum PSI. A major limitation of their protocol is its dependence on the assumption that the majority of parties are honest.
Ma et al. \cite{ma2024over} presented a novel OT-MP-PSI by introducing the dual cloud framework. In this design, the clients only need to pre-process the data and delegate the subsequent computation to cloud servers, which substantially reduces both the computational and communication overhead on the clients.
Yang et al. \cite{yang2024efficient} proposed the first unbalanced T-MP-PSI based on fully homomorphic encryption. Their construction achieves logarithmic communication complexity in the semi-honest setting, thereby offering a significant improvement in efficiency compared with previous work.

\section{Conclusion}
Most MP-PSI protocols with threshold, being fully anonymous, are often unsuitable for regulatory scenarios.
Moreover, the existing related scheme with traceability exhibits limitations in terms of both security and performance.
This paper introduces two novel Traceable Over-Threshold
Multi-Party Private Set Intersection (T-OT-MP-PSI) protocols to address more flexible privacy-preserving set intersection challenges. The first protocol Efficient T-OT-MP-PSI leverages OPPRF and Shamir's secret sharing to achieve high efficiency in the semi-honest model, ensuring resilience against up to $t-2$ colluding participants. The second protocol Security-enhanced T-OT-MP-PSI enhances security by introducing the oblivious linear evaluation  protocol, improving its ability to resist collusion by up to $n-1$ participants. Experimental results demonstrate the practicality and strong performance of both protocols, showing significant improvement over the existing solution. For instance, with 5 participants, a threshold value of 3, and set sizes of $2^{14}$, our protocols are \textcolor{black}{15056$\times$} (ET-OT-MP-PSI) and \textcolor{black}{505$\times$} (ST-OT-MP-PSI) faster than the work of Mahdavi et al., respectively.


\section*{Acknowledgment}
We thank the anonymous reviewers for their thoughtful comments. This work was supported in part by the National Natural Science Foundation of China Youth Project (No.62202102),  Scientific and Technological Project of Fujian Province of China (No.2024J08162),  the CCF-NSFOCUS `Kunpeng' Research Fund (No.CCF-NSFOCUS2024004),  National Key RD Plan of China (2020YFB1005803) and Key RD Plan of Shandong Province (2020CXGC010115).



%
\bibliographystyle{IEEEtran}
\bibliography{ref}

\section*{\scshape Appendix}

\addcontentsline{toc}{section}{APPENDIX}
\renewcommand{\thesubsection}{\Alph{subsection}.}

\subsection{Correctness and Security Proofs of ET-OT-MP-PSI}\label{PROOF OF ET-OT-MP-PSI}
\noindent\textbf{Theorem 1.}
\textit{The ET-OT-MP-PSI realizes the functionality $\mathcal{F}_{\text{T-OT-MP-PSI}}^{n,m,t}$ and is secure against collusion among up to $t - 2$ parties in the semi-honest model, given the statistical security parameter $\lambda$ and the computational security parameter $\kappa$.}
\subsubsection{Correctness}
\begin{proof}
    To analyze the correctness of the protocol, we consider the following two cases for each element $e_k^0$ in the set of $P_0$:
    \begin{itemize}
        \item \textbf{Case 1:} $ e_k^0 $ is an element of the intersection, i.e., it is held by at least $t$ parties, including $P_0$.
        \item \textbf{Case 2:} $ e_k^0 $ is not an element of the intersection, meaning that fewer than $t$ parties hold this element.
    \end{itemize}  
    \noindent\underline{\textit{Case 1: $ e_k^0 $ is an element of the intersection.}}   
    \begin{enumerate}
        \item \textbf{Conditional secret sharing.} $P_0$ treats the element $e_k^0$ as a secret and applies Shamir's secret sharing scheme to generate $n$ shares $s_k^{0,0}, \ldots, s_k^{0,n-1}$. These shares are then conditionally delivered to the other parties via OPPRF. Since $e_k^0$ is an element of the intersection, at least $t$ parties will receive the correct shares, i.e., the OPPRF output satisfies $\hat{s}_k^{0,i} = s_k^{0,i}$ for those parties.       
        \item \textbf{Secret shares update.} 
        For the $b^{\text{th}}$ bin, every party $P_i$, where $i \in [1, n-1]$, generates a polynomial $f_{i,b}(\cdot)$ and sends the evaluation $f_{i,b}(j+1)$ to each other party $P_j$, where $j \in [n]$. Each party $P_j$ then adds up the $n-1$ received values to compute
        $
        \delta_b = f_{1,b}(j+1) + \dots + f_{n-1,b}(j+1).
        $
        The secret share is then updated as $\mu_k^{0,i} = \hat{s}_k^{0,i} + \delta_b$. Since at least $t$ parties have received correct shares, at least $t$ of the updated shares are also correct.
        \item \textbf{Conditional collection and reconstruction.} 
        Each party $P_i$, where $i \in [1, n-1]$, invokes OPPRF protocol with $P_0$ to conditionally deliver the updated share $\mu_k^{0,i}$. Since the element $e_k^0$ is held by at least $t$ parties, $P_0$ obtains at least $t$ correct shares from the OPPRF, denoted as $y_k^i$, where $y_k^i = \mu_k^{0,i}$. Then, $P_0$ attempts to reconstruct the original secret by selecting $t$ values from the $n$ received outputs. When the selected $t$ shares are all correct, the original secret can be reconstructed, i.e., ${Recon}(y_k^i) = e_k^0$. Therefore, $e_k^0$ can be identified as an element in the intersection, and the parties holding correct shares are the holders of this intersection element.
    \end{enumerate}    
   \noindent\underline{\textit{Case 2:$ e_k^0 $ is not an element of the intersection.}}
    \begin{enumerate}
        \item \textbf{Conditional secret sharing.}
        The procedure is the same as in Case~1. However, since the element $e_k^0$ is not an element of the intersection—i.e., it is held by fewer than $t$ parties—the number of correct secret shares obtained by $P_0$ through the OPPRF protocol is less than $t$.
        \item \textbf{Secret shares update.} 
        The execution process is identical to that of Case~1. Since fewer than $t$ parties obtained correct shares in phase~1, the number of correct updated shares after the shares update phase remains less than $t$.
        \item \textbf{Conditional collection and reconstruction.} 
        The number of correct shares that $P_0$ obtains from the OPPRF outputs is less than $t$. When $P_0$ attempts to reconstruct the secret, it fails to recover the original secret with overwhelming probability under the given security parameters $\lambda$ and $\kappa$. As a result, $P_0$ determines that the element $e_k^0$ is not part of the intersection.
    \end{enumerate}
\end{proof}

\subsubsection{Security}
\begin{proof}
    As discussed in Section \ref{sub:Security-Enhanced Traceable OT-MP-PSI}, the ET-OT-MP-PSI is insecure in the presence of collusion among $ t-1 $ corrupted parties. Therefore, according to Definition 1, we prove that it is secure in the presence of collusion among up to $t - 2$ semi-honest adversaries. The following two distinct collusion scenarios should be considered. Here, we assume that Shamir's secret sharing and the OPPRF operate over the same field \( \mathbb{F}_p \).
    Let $\mathbb{C}$ and $\mathbb{H}$ be a coalition of corrupt and honest participants respectively.    And let $X$ and $Y$ be the inputs and outputs of the the coalition $ \mathbb{C} $.
    \begin{itemize}
        \item \textbf{Case 1:} Party $P_0$ is honest, and $ t-2 $ other parties are colluding, i.e. $\mathbb{C}\subseteq \{P_1, \dots, P_{n-1}\},|\mathbb{C}| = t - 2$.
        \item \textbf{Case 2:} Party $ P_0 $ is corrupted and colludes with $ t-3 $ other parties, i.e. $\mathbb{C} = \{P_0\}\cup \mathbb{C}_1$, where $\mathbb{C}_1\subseteq\{P_1, \dots, P_{n-1}\},|\mathbb{C}_1|=t-3$.
    \end{itemize}
    \noindent\underline{\textit{Case 1: Party $P_0$ is honest.}}
    
    In this case, the simulator $\mathsf{Sim}$ is given the inputs $X$ and outputs $Y=\perp$ of the corrupted parties $\mathbb{C}$, and runs as follows:
     \begin{enumerate}
        \item \textbf{Conditional secret sharing.} 
        $\mathsf{Sim}$ samples $t-2$ random values $\overline{\hat{s}_k^{0,i}}\leftarrow\mathbb{F}_p$, as the programmed outputs of the OPPRF.
        \item \textbf{Secret shares update.} 
        For the $b^{\text{th}}$ bin, $\mathsf{Sim}$ randomly generates polynomial $\overline{f_{j,b}(\cdot)}$ with a constant term of 0 and a degree of at most $t-1$, where $j\in\mathbb{H}$. Then, $\mathsf{Sim}$ computes $\overline{\delta_b}=\sum_{j\in \mathbb{C}}f_{j,b}(i+1) \ + \sum_{j\in \mathbb{H}}\overline{f_{j,b}}(i+1)$ and $\overline{y_i^k}= \overline{\hat{s}_k^{0,i}}+\overline{\delta_b}$.
        \item \textbf{Conditional collection and reconstruction.}
        At this stage, the corrupted parties receive no input. Therefore, the simulator $\mathsf{Sim} $  generates their corresponding views by faithfully following the protocol steps.
    \end{enumerate}

    Given the statistical security parameter $\lambda$ and the computational security parameter $\kappa$, both the OPPRF and our secret sharing scheme operate over the same finite field $\mathbb{Z}_p$. Now we argue that the views generated by $\mathsf{Sim}$ are computationally indistinguishable from those in the real execution.
    \begin{itemize}
        \item In the real world, according to the obliviousness of OPPRF, each $\hat{s}_k^{0,i}$ is indistinguishable from $\overline{\hat{s}_k^{0,i}}$ for OPPRF receiver, i.e. $\hat{s}_k^{0,i} \overset{c}{\equiv} \overline{\hat{s}_k^{0,i}} $. So the simulated views are computationally indistinguishable from the views in the real execution.
        \item During the real execution, the honest parties generate polynomials ${f_{j,b}(\cdot)}$ randomly, and each evaluation ${f_{j,b}(i+1)} \leftarrow\mathbb{F}_p$. Similarly, the simulator generates values $\overline{f_{j,b}(i+1)} \leftarrow\mathbb{F}_p$, which means ${f_{j,b}(i+1)}  \overset{c}{\equiv} \overline{f_{j,b}(i+1)}$ and hence ${\delta_b} \overset{c}{\equiv}\overline{\delta_b}$.
    \end{itemize}
    
    Therefore, we have $\{\mathsf{Sim}(X,Y,\mathbb{C})\} \overset{c}{\equiv} \{ \mathsf{view}^\pi_\mathbb{\mathbb{C}}(X,Y)\}$.

    \noindent\underline{\textit{Case 2: Party $P_0$ is corrupted.}}
    
    In this case, colluding parties learn the final outputs $I$ of the protocol. $\mathsf{Sim}$ is given the inputs $X$ and outputs $Y=I$ of the corrupted parties $\mathbb{C}$, and runs as follows:
    \begin{enumerate}
        \item \textbf{Conditional secret sharing.} 
        $\mathsf{Sim}$ selects $t-2$ random values $\overline{\hat{s}_k^{0,i}}\leftarrow\mathbb{F}_p$, as the outputs of the OPPRF.

        \item \textbf{Secret shares update.} This process is identical to Case 1: the simulator $\mathsf{Sim}$ generates the simulated polynomial $\overline{f_{j,b}(\cdot)}$ and computes $\overline{\delta_b}$.

        \item \textbf{Conditional collection and reconstruction.} If $ e_k^0 \notin I $, or $e_k^0 \in I$ but $e_k^0 \notin S_i$, $\mathsf{Sim}$ selects random values as the OPPRF outputs $\overline{y_k^i}$ received by $ P_0 $ from the honest party $P_i$. If $ e_k^0 \in I $ and $e_k^0 \in S_i$, the simulator $\mathsf{Sim} $ computes the correct outputs of the OPPRF protocol:
        $$
        \overline{y_k^i} = s_k^{0,i} + \sum_{P_j \in H} \overline{f_{j,b}(i+1)} + \sum_{P_j \in C} f_{j,b}(i+1).
        $$
        
    \end{enumerate}
    
    Given the statistical security parameter $\lambda$ and the computational security parameter $\kappa$, both the OPPRF and our secret sharing scheme operate over the same finite field $\mathbb{Z}_p$. Now we argue that the views generated by $\mathsf{Sim}$ are computationally indistinguishable from those in the real execution.
    
    \begin{itemize}
        \item According to the obliviousness of OPPRF, the receiver cannot distinguish between the simulated value $\overline{\hat{s}_k^{0,i}}$ and the actual value $\hat{s}_k^{0,i}$ generated during the real execution of the protocol. Therefore, we have $\hat{s}_k^{0,i} \overset{c}{\equiv} \overline{\hat{s}_k^{0,i}}$.
        \item Similar to Case 1, $\mathsf{Sim}$ samples random polynomial $\overline{f_{j,b}(\cdot)}$ from the same distribution over $\mathbb{F}_p$ as the honest parties do in the real execution. Therefore, we have $f_{j,b}(i+1) \overset{c}{\equiv} \overline{f_{j,b}(i+1)}$ and $ \delta_b \overset{c}{\equiv} \overline{\delta_b}$.
        \item $\mathsf{Sim}$ is given the final output $ I $ of the protocol. From this, $\mathsf{Sim}$ can determine the correct OPPRF outputs for each element in the intersection.  
        If $ e_k^0 \notin I $, or $ e_k^0 \in I $ but $ e_k^0 \notin S_i $, then by the obliviousness of OPPRF, the random value $\overline{y_k^i}$ chosen by $\mathsf{Sim}$ satisfies  
        $
        \overline{y_k^i}\overset{c}{\equiv} \text{output}_{\mathsf{OPPRF}}.
        $
        If $ e_k^0 \in I $ and $ e_k^0 \in S_i $, then $\mathsf{Sim}$ can compute the correct OPPRF output.  
        Therefore, the simulated value $ \overline{y_k^i} $ satisfies  
        $
        y_k^i \overset{c}{\equiv} \overline{y_k^i}.
        $
    \end{itemize}
     
     Therefore, we have $\{\mathsf{Sim}(X,Y,\mathbb{C})\} \overset{c}{\equiv} \{ \mathsf{view}^\pi_\mathbb{\mathbb{C}}(X,Y)\}$.
\end{proof}

\subsection{Correctness and Security Proofs of ST-OT-MP-PSI}\label{PROOF OF ST-OT-MP-PSI}
\noindent\textbf{Theorem 2.}
\textit{The ST-OT-MP-PSI realizes the functionality $\mathcal{F}_{\text{T-OT-MP-PSI}}^{n,m,t}$ and is secure against collusion among up to $n-1$ parties in the semi-honest model, given the statistical security parameter $\lambda$ and the computational security parameter $\kappa$. } 

\subsubsection{Correctness}
\begin{proof}
     Similarly, for each element $e_k^0$ in the set of $P_0$, we analyze the following two distinct cases:

    \begin{itemize}
        \item \textbf{Case 1:} $ e_k^0 $ is an element of the intersection, i.e., it is held by at least $t$ parties, including $P_0$.

        \item \textbf{Case 2:} $ e_k^0 $ is not an element of the intersection, meaning that fewer than $t$ parties hold this element.
    \end{itemize}

    \noindent\underline{\textit{Case 1: $ e_k^0 $ is an element of the intersection.}}
    
    \begin{enumerate}
        \item \textbf{Conditional secret sharing.} 
        The protocol proceeds in essentially the same way as in ET-OT-MP-PSI, except that the shared secret is a random value denoted by $e_k^{0'}$. In this phase, at least $t$ parties obtain correct secret shares.

        \item \textbf{Secret shares update.} 
        Each party $P_j$, where $j \in [1, n-1]$, generates a polynomial $f_{j,b}(\cdot)$ and directly sends $f_{j,b}(1)$ to $P_0$. $P_0$ then uses this value to correctly update its share. Subsequently, $P_0$ invokes OLE protocol with $P_j$, in which $P_j$ obtains the following OLE output:
        $
        z_0^v = (r_1^v + \cdots + r_{n-1}^v) \cdot e_k^0 + (a_1^{0,v} + \cdots + a_{n-1}^{0,v}).
        $
        Similarly, each party $P_i$, where $i \in [1, n-1]$, invokes OLE protocol with $P_j$ and obtains the output:
        $
        z_1^v = - (r_1^v + \cdots + r_{n-1}^v) \cdot e_k^i + (a_1^{1,v} + \cdots + a_{n-1}^{1,v}).
        $
        $P_i$ then uses $z_1^v$ to update its share as follows:
        $
        \mu_k^{0,i} = \hat{s}_k^{0,i} + z_1^v.
        $
        \item \textbf{Conditional collection and reconstruction.} 
        Each party $P_i$, where $i \in [1, n-1]$, invokes OPPRF protocol with $P_0$ to conditionally deliver the updated share $\mu_k^{0,i}$. Since the element $e_k^0$ is held by at least $t$ parties, $P_0$ receives at least $t$ correct updated shares, i.e., $y_k^i = \mu_k^{0,i}$ for those parties.
        Then, $P_0$ adds each received value $y_k^i$ to the corresponding value $z_1^0$ to obtain the final updated share: $y_k^i = y_k^i + z_1^0$. Since at least $t$ out of the $n$ updated shares obtained by $P_0$ are correct, $P_0$ can successfully reconstruct the original secret by using these correct shares, i.e., ${Recon}(y_k^i) = e_k^{0'}$. Therefore, $P_0$ can correctly identify that $e_k^0$ is in
        the intersection, and the parties corresponding to the correct shares are the holders of this intersection element.

    \end{enumerate}

    \noindent\underline{\textit{Case 2:$ e_k^0 $ is not an element of the intersection.}}

     \begin{enumerate}
        \item \textbf{Conditional secret sharing.} 
        The process is identical to that of Case~1. However, since fewer than $t$ parties hold the element $e_k^0$, fewer than $t$ correct secret shares are obtained after the OPPRF.

        \item \textbf{Secret shares update.} 
        The procedure is exactly the same as in Case~1.

        \item \textbf{Conditional collection and reconstruction.}
         After the OPPRF execution, the number of correct updated shares obtained by $P_0$ is fewer than $t$. When $P_0$ attempts to reconstruct the secret using $n$ received values, it fails to recover the original secret with overwhelming probability under the given security parameters $\lambda$ and $\kappa$. Therefore, $P_0$ concludes that $e_k^0$ is not an element of the intersection.

    \end{enumerate}
\end{proof}
\subsubsection{Security}
\begin{proof}
    To prove the security, we consider the following two cases. Both Shamir's secret sharing and the OPPRF are assumed to operate over the same finite field \( \mathbb{F}_p \).
    Let $ \mathbb{C} $ and $ \mathbb{H} $ denote the colluding parties and honest parties, respectively.
    \begin{itemize}
        \item \textbf{Case 1:} Party $P_0$ is honest, and other parties are colluding.

        \item \textbf{Case 2:} Party $ P_i $ is honest, where $ i \in [1, n-1] $, while the remaining parties, including $ P_0 $, are corrupted.

    \end{itemize}

    \noindent\underline{\textit{Case 1: Party $P_0$ is honest.}}

     In this case, the simulator $\mathsf{Sim}$ is given the inputs $X$ and outputs $Y=\perp$ of the corrupted parties $\mathbb{C}$, and runs as follows:
     
    \begin{enumerate}
        \item \textbf{Conditional secret sharing.}  $\mathsf{Sim}$ samples $n-1$ random values $\overline{\hat{s}_k^{0,i}}\leftarrow\mathbb{F}_p$, as the simulated outputs of the OPPRF.

        \item \textbf{Secret shares update.} $\mathsf{Sim}$ selects random values $\overline{e_k^0}\leftarrow\mathbb{F}_p$ as the simulated inputs of $P_0$ and the corrupted parties for the OLE protocol.

        \item \textbf{Conditional collection and reconstruction.} As the corrupted parties obtain no inputs during this phase, $\mathsf{Sim}$ can simulate their views according to the protocol, resulting in views that are computationally indistinguishable from the real ones.
    \end{enumerate}

    Given the statistical security parameter $\lambda$ and the computational security parameter $\kappa$, both the OPPRF and our secret sharing scheme operate over the same finite field $\mathbb{Z}_p$. Now we argue that the views generated by $\mathsf{Sim}$ are computationally indistinguishable from those in the real execution.

    \begin{itemize}
        \item First, since both the OPPRF and Shamir's secret sharing operate over the same field $\mathbb{F}_p$, and due to the obliviousness of OPPRF, the value $\overline{\hat{s}_k^{0,i}}$ chosen by $\mathsf{Sim}$ is computationally indistinguishable from the real value $\hat{s}_k^{0,i}$ from the receiver's perspective. That is, $\hat{s}_k^{0,i} \overset{c}{\equiv} \overline{\hat{s}_k^{0,i}}$.

        \item In the real world, $P_0$'s inputs to the OLE are the set elements $e_k^0$, sampled uniformly from $\mathbb{F}_p$. Therefore, $e_k^0 \overset{c}{\equiv} \overline{e_k^0}$.
    \end{itemize}

     Therefore, we have $\{\mathsf{Sim}(X,Y,\mathbb{C})\} \overset{c}{\equiv} \{ \mathsf{view}^\pi_\mathbb{\mathbb{C}}(X,Y)\}$.

    \noindent\underline{\textit{Case 2: Party $P_0$ is corrupted.}}

    In this case, colluding parties learn the final outputs $I$ of the protocol. $\mathsf{Sim}$ is given the inputs $X$ and outputs $Y=I$ of the corrupted parties $\mathbb{C}$, and runs as follows:
    
    \begin{enumerate}
        \item \textbf{Conditional secret sharing.} 
        $\mathsf{Sim}$ samples $n-1$ random values $\overline{\hat{s}_k^{0,i}}\leftarrow\mathbb{F}_p$, as the outputs of the OPPRF.

        \item \textbf{Secret shares update.} For the $b^{\text{th}}$ bin, $\mathsf{Sim}$ constructs a random polynomial $\overline{f_{i,b}(\cdot)}$ with a constant term of 0 and a degree of at most $t-1$, and sets $\overline{f_{i,b}(1)}$ as the value used by the corrupted party $P_0$ to update its share.  
        If the element held by the corrupted party $P_0$ is $e_0$ and the element held by another corrupted party $P_j$ is $e_1$, $\mathsf{Sim}$ generates the simulated OLE outputs $\overline{c_0} = \overline{r_i} \cdot e_0 + \overline{a_i^0}$ for $P_0$ and $\overline{c_1} = -\overline{r_i} \cdot e_1 + \overline{a_i^1}$ for $P_j$, where $\overline{a_i^0} + \overline{a_i^1} = \overline{f_i(j+1)}$ and $\overline{r_i}$ are random values.

        \item \textbf{Conditional collection and reconstruction.} If $ e_k^0 \notin I $, or $ e_k^0 \in I $ but $e_k^0 \notin S_i$, $\mathsf{Sim}$ generates a random value to simulate the OPPRF output from the honest party $P_i$. If $ e_k^0 \in I$ and $e_k^0 \in S_i$, the simulator deduces the correct OPPRF output:      
        $$
            \overline{y_k^i} = s_k^{0,i}  -(\sum_{j=1,j\neq i}^{n-1} r_j +\overline{r_i})\cdot e_k^0 + 
            ( \sum_{j = 1,j\neq i}^{n-1} a_j^1 + \overline{a_i^1} ).
        $$
        Similarly, if $e_k^0 \notin I$, or $e_k^0 \in I$ but $e_k^0 \notin S_i$, then $\mathsf{Sim}$ selects a random value $\overline{v'}$ as the OLE index derived from the OPPRF.  
        If $e_k^0 \in I$ and $e_k^0 \in S_i$, then $\mathsf{Sim}$ can compute the correct OLE index.

    \end{enumerate}

    Given the statistical security parameter $\lambda$ and the computational security parameter $\kappa$, both the OPPRF and our secret sharing scheme operate over the same finite field $\mathbb{Z}_p$. Now we argue that the views generated by $\mathsf{Sim}$ are computationally indistinguishable from those in the real execution.

    \begin{itemize}
        \item According to the obliviousness of OPPRF, each $\hat{s}_k^{0,i}$ is indistinguishable from $\overline{\hat{s}_k^{0,i}}$ to the receiver, i.e. $\hat{s}_k^{0,i} \overset{c}{\equiv} \overline{\hat{s}_k^{0,i}} $.
        \item In the real world, polynomials ${f_{j,b}(\cdot)}$ are generated randomly, and each evaluation ${f_{j,b}(i+1)} \leftarrow\mathbb{F}_p$. Similarly, $\mathsf{Sim}$ generates values $\overline{f_{j,b}(i+1)} \leftarrow\mathbb{F}_p$, which means ${f_{j,b}(i+1)}  \overset{c}{\equiv} \overline{f_{j,b}(i+1)}$.
        According to the obliviousness of OLE, the receiver cannot computationally distinguish the OLE output from a uniformly random value over the same field $\mathbb{F}_p$. Therefore, the simulated outputs generated by $\mathsf{Sim}$ are computationally indistinguishable from those in the real execution, i.e., $c_0 \overset{c}{\equiv} \overline{c_0}$ and $c_1 \overset{c}{\equiv} \overline{c_1}$.

        \item $\mathsf{Sim}$ is given the final output $ I $ of the protocol. From this, $\mathsf{Sim}$ can determine the correct OPPRF outputs for each element in the intersection. If $ e_k^0 \notin I $, or $ e_k^0 \in I $ but $ e_k^0 \notin S_i $, then by the obliviousness of OPPRF, we have $ y_k^i \overset{c}{\equiv} \overline{y_k^i}$ and $ v' \overset{c}{\equiv} \overline{v'}$. If $ e_k^0 \in I $ and $ e_k^0 \in S_i $, then $\mathsf{Sim}$ can compute the correct OPPRF output. Consequently, $ y_k^i \overset{c}{\equiv} \overline{y_k^i} $ and $ v' \overset{c}{\equiv} \overline{v'}$.

    \end{itemize}

    Therefore, we have $\{\mathsf{Sim}(X,Y,\mathbb{C})\} \overset{c}{\equiv} \{ \mathsf{view}^\pi_\mathbb{\mathbb{C}}(X,Y)\}$.
    
\end{proof}

\end{document}